\documentclass[graybox]{svmult}

\usepackage{type1cm}        
\usepackage{makeidx}         
\usepackage{graphicx}        
\usepackage{multicol}        
\usepackage[bottom]{footmisc}

\usepackage{newtxtext}       %
\usepackage[varvw]{newtxmath}       

\makeindex             

\usepackage{float}
\usepackage{algpseudocode}
\usepackage{tikz}
\usetikzlibrary{positioning}
\usepackage{caption}
\usepackage{changepage}
\newcommand*\diff{\mathop{}\!\mathrm{d}}
\makeatletter
\newcommand{\setword}[2]{%
  \phantomsection
  #1\def\@currentlabel{\unexpanded{#1}}\label{#2}%
}

\title*{Convergence of the Euler--Maruyama particle scheme for a regularised McKean--Vlasov equation arising from the calibration of local-stochastic volatility models}
\titlerunning{Convergence of Euler--Maruyama particle scheme for regularised calibration problem}

\author{Christoph Reisinger and Maria Olympia Tsianni}

\institute{Christoph Reisinger  \at Mathematical Institute, Oxford University, \email{christoph.reisinger@maths.ox.ac.uk}
\and Maria Olympia Tsianni \at Mathematical Institute, Oxford University, \email{maria.tsianni@maths.ox.ac.uk}}

\begin{document}

\maketitle

\abstract{In this paper, we study the Euler--Maruyama scheme for a particle method to approximate the McKean--Vlasov dynamics of calibrated local-stochastic volatility (LSV) models. Given the open question of well-posedness of the original problem,
we work with regularised coefficients and prove that  under certain assumptions on the inputs, the regularised model is well-posed. Using this result, 
we prove the strong convergence of the Euler--Maruyama scheme to the particle system with rate $1/2$ in the step-size and obtain an explicit dependence of the error on the regularisation parameters.
Finally, we implement the particle method for the calibration of a Heston-type LSV model to illustrate the convergence in practice and to investigate how the choice of regularisation parameters affects the accuracy of the calibration.}

\abstract*{In this paper, we study the Euler--Maruyama scheme for a particle method to approximate the McKean--Vlasov dynamics of a calibrated local-stochastic volatility (LSV) model. Given the open question of whether an LSV model exists for any given arbitrage-free smile, we work on a regularised setting of the calibrated dynamics and prove that for fixed regularisation parameters and certain assumptions on the inputs of the model, the regularised problem is well-posed. Using this result, 
we prove the strong convergence of the Euler--Maruyama scheme on the particle system with rate $1/2$ in the step-size and obtain an explicit dependence between the error and the regularisation parameters.
Finally, we implement the particle method for the calibration of a Heston-type LSV model to investigate how the choice of regularisation parameters affects the accuracy of the calibration.}


\section{Introduction}
Since the Black--Scholes (BS) model was first introduced, extensive research in quantitative finance has taken place for the development of more sophisticated models that successfully price and hedge financial instruments. An extension to the BS model is the Local Volatility (LV) model introduced by Dupire in \cite{DUPIREE}. The LV model exactly reproduces any arbitrage free volatility surface, however it has unrealistic dynamics. Another class of more enhanced models are the Stochastic Volatility (SV) models, which generate an implied volatility smile and better describe the market dynamics. However, as parametric models, they only have a finite number of parameters and are therefore unable to capture the entire implied volatility surface. 
Local-stochastic volatility models (LSV), first published to our knowledge in \cite{jex}, combine the strengths of both LV and SV models and are the \textit{state-of-the-art} in the finance industry.

For a given time horizon $[0,T]$, a general LSV model is of the form 
\begin{equation} \label{general}
    \diff S_t = S_t\, g(Y_t)\, \sigma(t,S_t)\, \diff W_t, 
\end{equation}
where $S_t$ is the current value of the one-dimensional process $S = (S_t)_{t\in [0,T]}$, the spot price of the underlying asset, and $Y_t$ is the current value of the stochastic volatility process $Y = (Y_t)_{t\in[0,T]}$. Popular choices for $(Y_t)_{t \in [0,T]}$ are the exponential Ornstein--Uhlenbeck and the Cox--Ingersoll--Ross processes. Through the stochastic volatility component, $g(Y_t)$, the model better captures stylised features of the market dynamics such as volatility clustering and negative correlation between asset price and volatility. Embedding the local volatility $\sigma(t,S_t)$ brings accuracy to the model as it can exactly calibrate any market observed volatility surface. Indeed, an LSV model is able to better price and hedge options, as studied by several practitioners in the field throughout the years, including the early review by Piterbarg, \cite{piterbarg}. 

In our work, we consider the SDE \eqref{general} describing the process $(S_t)_{t\ge0}$
under the risk-neutral measure $(\mathbb{Q}_t)_{t\ge0}$, which supports a two-dimensional Brownian motion $(W^s,W^y)$ with $\diff W^s_t \diff W^y_t = \rho_{S,Y} \diff t$, and where $(Y_t)_{t\ge0}$ is given by
\begin{eqnarray}
    \diff Y_t &=& m(\theta - Y_t)\diff t + \gamma \diff W^y_t,
\end{eqnarray}
with parameters $m>0, \theta, \gamma$.


A consistency condition for exact calibration of such models can be derived from Gy\"ongy's result \cite{mimicking}, 
as given by Dupire for general stochastic volatility models in \cite{iff} and specified for SLV in \cite{jex},
\begin{equation}\label{Dupireiff}
\sigma^{2}(t, S)=
\frac{\sigma_{\text{Dup}}^{2}(t,S)}{\mathbb{E}^{\mathbb{Q}}[g^2(Y_{t})|S_{t} = S]},
\end{equation}
where $\sigma_{\text{Dup}}(t,S)$ denotes the local volatility and is given by the Dupire formula 
\begin{equation*}
\sigma_{\text{Dup}}^{2}(T,S) = \frac{\frac{\partial C(T,S)}{\partial T} }{\frac{S^2}{2} \frac{\partial^2 C(T,S)}{\partial K^2} },
\end{equation*}
for given call option prices with maturity $T$ and strike $K$ observed in the market, assuming zero interest and dividend rates for simplicity.

The conditional expectation $\mathbb{E}^{\mathbb{Q}}[g^2(Y_{t})|S_{t} = S]$ creates a dependence of the diffusion coefficient of $(S_t)_{t\ge 0}$ on the underlying joint distribution of $(S_t,Y_t)$ and therefore leads to a McKean\textendash Vlasov SDE (see McKean's seminal work \cite{MKV}). 
This nonlinear law dependence, and the presence of a conditional expectation in particular,
renders the SDE challenging and has led to the development of sophisticated calibration techniques, mainly in two different directions. One is the particle method, introduced for this problem by Guyon and Henry-Labord{\`e}re in \cite{book}, Chapter 11, and the other is the PDE approach, which is based on the solution of the Fokker--Planck equation as in \cite{lipton} and \cite{ren}. Both of these methods require a priori knowledge of the local volatility surface, which can be calculated using the Dupire formula. Since there is only a finite number of options available in the market, then ad hoc interpolation of the option prices or the volatility surface is necessary, which can however lead to instabilities and inaccuracies, as explained in \cite{guo}. In the more recent study \cite{CUCHIERO}, Cuchiero et al.\ calibrate LSV models using deep learning. Specifically, the authors use a set of feed-forward neural networks to parameterise the leverage function and calibrate the model using a generative adversarial network approach so that they avoid traditional interpolation. In this paper, we focus on the Monte Carlo particle method as in \cite{book}.

Although LSV models are very powerful tools in pricing and risk-management, existence and uniqueness of a solution to the calibrated LSV model have not been established to date. The main challenge arises from the leverage function that appears in the diffusion coefficient of the calibrated dynamics, which involves the conditional expectation of a function of the volatility given the value of the process $S$ and therefore makes the equation nonlinear and nonlocal. The problem has been attempted by several researchers in the field, however only partial results so far exist. In \cite{abergeltachet}, Abergel and Tachet  
prove that under certain assumptions and regularisation of the initial condition, and for small enough volatility of volatility $g(\cdot)$, a related initial-boundary value problem is well-posed up to a finite maturity $T^{*} \le T$. 
Lacker et al.\ show in \cite{lacker} existence and uniqueness of solutions of the type \eqref{general} in the stationary case.
Another result from Jourdain and Zhou in \cite{jourdain} proves existence in the case when the stochastic volatility component is a jump process with a finite number of states. In the more recent study \cite{newpaper}, Djete uses Sobolev estimates to prove existence (and a propagation of chaos result) in a class of McKean--Vlasov equations with weak continuity assumptions in the measure variable and deduces the existence of a calibrated LSV model.

In this paper, we work with a regularised formulation of the calibrated dynamics \eqref{general} and prove that under certain assumptions the regularised equation is well-posed. We recently became aware of a related regularisation approach by Bayer et al.\ in \cite{bayer}, employing reproducing kernel Hilbert space (RKHS) techniques, which also gives well-posedness (and propagation of chaos, as in our work).

Equivalently to our initial problem formulation \eqref{general}, one may consider the dynamics of the process $(X_t)_{t\in [0,T]}$ = $(\text{log}(S_t))_{t\in [0,T]}$. By It\^{o}'s lemma, we get the SDE describing the dynamics of $X_t$ under the risk neutral measure ${\mathbb{Q}}_t$:\\
\begin{equation}\label{sde}
\begin{split}
    &\diff X_t = -\frac{1}{2}g^2(Y_t)\frac{\sigma^2_{\text{Dup}}(t,e^{X_t})}{\mathbb{E}^{\mathbb{Q}}[g^2(Y_{t})|X_{t}]}\diff t + g(Y_t)\frac{\sigma_{\text{Dup}}(t,e^{X_t})}{\sqrt{\mathbb{E}^{\mathbb{Q}}[g^2(Y_{t})|X_{t}]}}\diff W^x_t, 
\end{split}
\end{equation}
and $(W^x, W^y)$ a two-dimensional Brownian motion with $\diff W^x_t \diff W^y_t = \rho_{X,Y} \diff t$.

In Section \ref{chap:lipschitz}, we prove the well-posedness of a regularised equation. 
 In Section \ref{chap:particlemethod}, we present and analyse a particle method, as introduced by Guyon and Henry-Labord\'ere in \cite{book}. Thereafter, in Section \ref{chap:euler}, we apply the standard Euler--Maruyama scheme to the particle system and 
 prove its strong convergence using the results of the Lipschitz continuity of the drift and diffusion coefficients of the regularised SDE from Section \ref{chap:lipschitz}. We note that with the regularity results from Section \ref{chap:lipschitz}, we could deduce the convergence of the time-stepping scheme from results in the literature (see, e.g., \cite{doksasi}), however, our direct proof allows us to obtain the dependence of the error on the regularisation parameters explicitly. Finally, in Section \ref{chap:python}, we implement the particle method for the calibration of a Heston-type local volatility model and illustrate our results. The diagram below shows the different steps of convergence that we show in our work. The final convergence as $\epsilon \to 0$ is not analysed in this paper but is currently being explored.
\bigskip
 
\begin{adjustwidth*}{}{2em}
 \begin{tikzpicture}
\node [draw,
	minimum width=2cm,
	minimum height=1.2cm
]  (start){Discretised process $\hat{Z}_t^{i,N}$};

\node [draw,
	minimum width=2cm,
	minimum height=1.2cm,
    right=2.4cm of start
]  (discretised){Particle system $Z_t^{i,N}$};

\node [draw,
	minimum width=2cm, 
	minimum height=1.2cm,
	below=0.9cm of discretised
] (particle){Regularised calibrated dynamics $Z_{\epsilon,t}$};

\node [draw,
	minimum width=2cm, 
	minimum height=1.2cm, 
	below = 0.9cm of particle
]  (prices1) {Call option prices $C_{\epsilon}$};

\node [draw,
	minimum width=2cm, 
	minimum height=1.2cm, 
	left = 1.8cm of prices1
]  (prices2) {Call option market prices};

\draw[-stealth] (start.east) -- ++ (2.4,0) 
node[midway,above]{$\Delta t \to 0$, fixed $\epsilon$}
 node[midway,below]{Section 4};

\draw[-stealth] (discretised.south) -- (particle.north) 
	node[midway,left]{$N \to \infty$, fixed $\epsilon$,}node[midway,right]{Section 3};

\draw[dashed,->] (particle.south) -- (prices1.north) 
	node[midway]{Use the regularised equation to price call options};

\draw[dashed,->] (prices1.west) --(prices2.east) 
	node[midway,above]{$\epsilon \to 0$};
 \end{tikzpicture}
\end{adjustwidth*}

\section{Existence and uniqueness for a regularised 
equation}
\label{chap:lipschitz}


For a given $T >0$, let $(\Omega , \mathcal{F}, \mathbb{F} = (\mathcal{F}_t)_{t \in [0,T]}, \mathbb{P}) $
denote a complete filtered probability space where $\mathbb{F}$ is the augmented filtration of a standard multidimensional Brownian motion $(W_t)_{t \in [0,T]}$. Additionally, let $(\mathbb{R}^d,\langle \cdot,\cdot \rangle,|\cdot|)$ denote the $d$-dimensional Euclidean space and $|\cdot|$ the Hilbert--Schmidt norm, $\mathcal{P}(\mathbb{R}^d)$ denote the set of all probability measures on the measurable space $(\mathbb{R}^d, \mathcal{B}(\mathbb{R}^d))$, where $\mathcal{B}(\mathbb{R}^d)$ represents the Borel $\sigma$-field over $\mathbb{R}^d$, and $\mathcal{P}_{2}(\mathbb{R}^d)$ denote the subset of $\mathcal{P}(\mathbb{R}^d)$ with probability measures with finite second moment so that 
$$\mathcal{P}_{2}(\mathbb{R}^d) := \{\mu \in \mathcal{P}(\mathbb{R}^d) \big| \int_{\mathbb{R}^d}|x|^{2}\mu(\diff x)< \infty \}.$$
Also, by $L_0^2(\mathbb{R})$ we denote the space of real-valued, $\mathcal{F}_0$-measurable random variables with finite second moments and by $\mathcal{S}^2([0,T])$ the space of $\mathbb{R}$-valued, $\mathbb{F}$-adapted continuous processes on $[0,T]$. 
The Wasserstein distance $\mathcal{W}_{2}(\mu , \nu)$ on $\mathcal{P}_{2}(\mathbb{R}^d)$ is 
\begin{equation}
    \mathcal{W}_{2}(\mu,\nu) := \inf_{\gamma \in \Gamma(\mu,\nu)}\bigg(\int_{\mathbb{R}^{d} \times \mathbb{R}^{d}} |x-y|^2 \gamma (\diff \mu,\diff \nu)\bigg)^{1/2},
\end{equation}
where $\Gamma(\mu,\nu)$  is the set of all couplings between $\mu$ and $\nu$, for $\mu, \nu \in \mathcal{P}(\mathbb{R}^d)$, such that $\gamma \in \Gamma(\mu,\nu)$ has marginals $\mu$ and $\nu$. 


We introduce a mollifier of the form 
\begin{equation} \label{kernel}
    \Phi_{\epsilon}(x) = \epsilon^{-1}K\left(\epsilon^{-1} x\right),
\end{equation}
where $K(\cdot)$ is a real-valued, non-negative kernel function with the normalization and symmetry properties $\int_{-\infty}^{+\infty} K(u)\diff u = 1 \text{ and } K(u) =  K(-u)\text{ } \forall \text{ } u$.
To approximate the conditional expectation and avoid potential singularities of the diffusion coefficient at 0, 
we apply a mollification and add a constant parameter $\delta$, as follows:
\begin{eqnarray*}
\tilde{\sigma}_{\epsilon}(t,(x,y),\mu) = g(y)\sigma_{\text{Dup}}(t,e^{x})\frac{\sqrt{\mathbb{E}^{\mu}[K(\frac{X_{\epsilon}-x}{\epsilon})]+ \delta}}{\sqrt{\mathbb{E}^{\mu}[g^{2}(Y_{\epsilon})K(\frac{X_{\epsilon}-x}{\epsilon})]+\delta}}, \quad
\tilde{b}_{\epsilon} = - \tilde{\sigma}_{\epsilon}^2/2.
\end{eqnarray*}
The system of processes that approximate the original SDE \eqref{sde} is therefore
\begin{equation}\label{sdeapprox}
\begin{split}
&\diff X_{\epsilon,t} =\tilde{b}_{\epsilon}(t,(X_{\epsilon,t},Y_{\epsilon,t}),\mu_t)\diff t+\tilde{\sigma}_{\epsilon}(t,(X_{\epsilon,t},Y_{\epsilon,t}),\mu_t)\diff W^x_t,\\
&\diff Y_{\epsilon,t} = m(\theta - Y_{\epsilon,t})\diff t + \gamma \diff W^y_t.
\end{split}
\end{equation}


\bigbreak
\noindent
\textbf{Assumptions A} \smallskip \newline
\setword{\textbf{A1.}}{Word:A1} \ The function $g(\cdot)$ is bounded and Lipschitz continuous so that for all $y_1,y_2 \in \mathbb{R}$, $|g(y_1)|\le A_1 $ and  $|g(y_1)-g(y_2)|\le L_g |y_1-y_2|$ for constants $A_1, L_g$.
\smallskip
\\
\setword{\textbf{A2.}}{Word:A2} \ The local volatility $(t,x) \to \sigma_{\text{Dup}}(t,e^{x})$ is bounded, Lipschitz in $x$, and $\frac{1}{2}$-H\"{o}lder  in $t$, so that for $A_2 \text{ and }L_{\text{Dup}}$ constants and for all $x_1,x_2 \in \mathbb{R}$ and $t_1,t_2 \in [0,T],$
$|\sigma_{\text{Dup}}(t_1,\cdot)|\le A_2, \text{ and } \big|\sigma_{\text{Dup}}(t_1,e^{x_1}) - \sigma_{\text{Dup}}(t_2,e^{x_2})\big| \le L_{\text{Dup}}\big(|t_1-t_2|^{1/2} + |x_1 - x_2| \big).$ \smallskip
\\
\setword{\textbf{A3.}}{Word:A3} \ The kernel function $K(\cdot)$ is bounded and Lipschitz continuous so that for all $x_1,x_2 \in \mathbb{R}$, $|K(x_1)|\le A_3, $ and  $|K(x_1)-K(x_2)|\le L_K |x_1-x_2|$, for constants $A_3, L_K$ .\smallskip \newline

Let $E^{\mu}_{f}(x) := \mathbb{E}^{\mu}[f^2(Y) K((X-x)/\epsilon)] + \delta$.\smallskip \newline

The following remarks are immediate from assumptions \ref{Word:A1}-\ref{Word:A3}
\begin{remark}\label{Word:Remark1}
\begin{equation} \label{lowerbound1}
\bigg |\sqrt{E^{\mu}_{g}(x_1)} \cdot \sqrt{E^{\nu}_{g}(x_1)} \bigg | \geq \delta
\end{equation}
\end{remark}
\bigbreak
\begin{remark}\label{Word:Remark2}
Using \ref{Word:A3},
\begin{equation} \label{upperbound}
   \text{(i) }\big|E^{\mu}_{1}(x_1)\big| = \big| \iint_{\mathbb{R}^2}K(\frac{x-x_1}{\epsilon}) \mu(\diff x,\diff y) +\delta\big|\\
   \le A_3\iint_{\mathbb{R}^2}\mu(\diff x,\diff y) +\delta  = A_3 +\delta.
\end{equation}
(ii) Similarly, by \ref{Word:A1} and \ref{Word:A3},
$\big| E^{\nu}_{g}(x_1)\big| \le A_1^2A_3\iint_{\mathbb{R}^2}|p(x,y)|\diff x\diff y+\delta = A_1^2A_3+\delta$.
\end{remark}

\begin{lemma}
\label{Remark3} 
There exists $M_1 >0$  such that
for any $\mu, \nu \in \mathcal{P}_{2}(\mathbb{R}^2)$, $\epsilon>0$,
and $x_1 \in \mathbb{R}$,\newline
$$\big| \mathbb{E^\mu}[g^2(Y)K(\frac{X-x_1}{\epsilon})] - \mathbb{E^\nu}[g^2(Y)K(\frac{X-x_1}{\epsilon})]\big| \le \frac{M_1}{\epsilon}\mathcal{W}_{2}(\mu,\nu).$$
\end{lemma}

\begin{proof} 
Let $\Gamma(\cdot,\cdot)$ denote an arbitrary coupling between $\mu(\cdot)$ and $\nu(\cdot)$ with $\Gamma(\mu,\nu)$ the set of all such couplings. 
Also, let $z := (x_z,y_z), w := (x_w,y_w) \in \mathbb{R}^2.$ Then
\begin{equation} \label{Wass1}
    \begin{split}
        &\hspace{0 cm} \big| \mathbb{E^\mu}[g^2(Y)K(\frac{X-x_1}{\epsilon})] - \mathbb{E^\nu}[g^2(Y)K(\frac{X-x_1}{\epsilon})]\big| \\
        &= \bigg|\iint_{\mathbb{R}\times \mathbb{R}}g^2(y)K(\frac{x-x_1}{\epsilon}) \mu(\diff x,\diff y) -  \iint_{\mathbb{R}\times \mathbb{R}}g^2(y)K(\frac{x-x_1}{\epsilon}) \nu(\diff x,\diff y) \bigg|\\
        &\le \iint_{\mathbb{R}^2 \times \mathbb{R}^2 } \bigg| g^2(y_z)K(\frac{x_z-x_1}{\epsilon}) - g^2(y_w)K(\frac{x_w-x_1}{\epsilon})\bigg| \Gamma(\diff z,\diff w). 
    \end{split}
\end{equation}

Let $f(x,y) := g^2(y)K(\frac{x-x_1}{\epsilon})$,
then by \ref{Word:A1} and \ref{Word:A3} there exists $M_1$ such that
$|f(z)-f(w)| \le (M_1/\epsilon) \, 
\big|z-w\big|$.
Substituting into \eqref{Wass1}, 
by Cauchy--Schwarz, 
\begin{eqnarray*} 
\big| \mathbb{E^\mu}[g^2(Y)K(\frac{X-x_1}{\epsilon})] - \mathbb{E^\nu}[g^2(Y)K(\frac{X-x_1}{\epsilon})]\big|   
    \le \frac{M_1 }{\epsilon}\bigg(\iint_{\mathbb{R}^2 \times \mathbb{R}^2 } \big|z-w\big|^2 \Gamma(\diff z,\diff w) \bigg)^{1/2}.
\end{eqnarray*}
Since the last bound holds for every coupling $\Gamma \in \Gamma(\mu,\nu)$, 
\begin{eqnarray*}
\big| \mathbb{E^\mu}[g^2(Y)K(\frac{X-x_1}{\epsilon})] &-& \mathbb{E^\nu}[g^2(Y)K(\frac{X-x_1}{\epsilon})]\big| \\   
&\le& \frac{M_1}{\epsilon} \bigg(\underset{{\Gamma} \in \Gamma_{\mu,\nu}}{\text{inf}}\iint_{\mathbb{R}^2 \times \mathbb{R}^2 } \big|z-w\big|^2 {\Gamma}(\diff z,\diff w) \bigg)^{1/2} = \frac{M_1}{\epsilon} \mathcal{W}_{2}(\mu,\nu),
\end{eqnarray*}
by the definition of the Wasserstein metric.
\end{proof}


We keep $\epsilon$ and $\delta$ fixed and only show that Lipschitz-continuity and linear growth conditions hold for $\tilde{\sigma}_{\epsilon}$, as the proof for $\tilde{b}_{\epsilon}$ follows from similar arguments.

\begin{proposition}
\label{condition}
Let $\tilde{b}_{\epsilon}: [0,T] \times \mathbb{R}^2  \times \mathcal{P}(\mathbb{R}^2) \xrightarrow{} \mathbb{R}$ and $\tilde{\sigma}_{\epsilon}: [0,T] \times \mathbb{R}^2  \times \mathcal{P}(\mathbb{R}^2) \xrightarrow{} \mathbb{R}$ be the drift and diffusion coefficients of process $X_{\epsilon}$ of equation \eqref{sdeapprox}. Under assumptions \ref{Word:A1}-\ref{Word:A3}, there exists a positive constant $L=O(\frac{1}{{\epsilon}\delta^2})$ such that $ \forall t, t_1, t_2 \in [0,T]$, $\forall (x,y), (x_1,y_1),(x_2,y_2) \in \mathbb{R}^2$, and $\forall \mu, \nu \in \mathcal{P}_2(\mathbb{R}^2)$, we have that
\begin{equation*} 
\begin{split}
\text{(i) }|\tilde{b}_{\epsilon}(t_1,(x_1,y_1), \mu)-&\tilde{b}_{\epsilon}(t_2,(x_2,y_2),\nu)|+|\tilde{\sigma}_{\epsilon}(t_1,(x_1,y_1), \mu)-\tilde{\sigma}_{\epsilon}(t_2,(x_2,y_2),\nu)|\\
\le& \, L \, \bigg(|t_1-t_2|^{1/2}+|x_1-x_2|+|y_1-y_2|+\mathcal{W}_{2}(\mu,\nu)\bigg)\\
\text{(ii) }|\tilde{b}_{\epsilon}(t,(x,y), \mu)|+&|\tilde{\sigma}_{\epsilon}(t,(x,y),\mu)|\le L(1+|x|+|y|).
\end{split}
\end{equation*}
\end{proposition}
\begin{proof}
$
|\tilde{\sigma}_{\epsilon}(t_1,(x_1,y_1) , \mu) - \tilde{\sigma}_{\epsilon}(t_2,(x_2,y_2) , \nu) | \le$
\begin{equation}\label{splitted1}
\begin{split}
  &
  \\
   & \le
   \overbrace{|g(y_1)\sigma_{\text{Dup}}(t_1,e^{x_1}) | \bigg|\frac{\sqrt{E^{\mu}_{1}(x_1)}}{\sqrt{E^{\mu}_{g}(x_1)}}  -  \frac{\sqrt{E^{\nu}_{1}(x_1)}}{\sqrt{E^{\nu}_{g}(x_1)}} \bigg|}^{=:T_1}\\
   & +
   \overbrace{|g(y_1)\sigma_{\text{Dup}}(t_1,e^{x_1}) | \bigg|  \frac{\sqrt{E^{\nu}_{1}(x_1)}}{\sqrt{E^{\nu}_{g}(x_1)}} -  \frac{\sqrt{E^{\nu}_{1}(x_2)}}{\sqrt{E^{\nu}_{g}(x_2)}} \bigg|}^{=:T_2}+\bigg|\sigma_{\text{Dup}}(t_1,e^{x_1})\frac{\sqrt{E^{\nu}_{1}(x_2)}}{\sqrt{E^{\nu}_{g}(x_2)}}\bigg| | g(y_1)- g(y_2)|\\
   &+\bigg|g(y_2)\frac{\sqrt{E^{\nu}_{1}(x_2)}}{\sqrt{E^{\nu}_{g}(x_2)}}\bigg | |\sigma_{\text{Dup}}(t_1,e^{x_1})- \sigma_{\text{Dup}}(t_2,e^{x_2})|.
\end{split}
\end{equation}

\textbf{First Term:} We show that for  $C_1$ a constant, $T_1 \le C_1\mathcal{W}_{2}(\mu,\nu)$. From \eqref{lowerbound1}, assumptions \ref{Word:A1} and \ref{Word:A2},
\begin{equation*}
\begin{split}
&
T_1 
\le \frac{A_1A_2}{\delta} \bigg |\sqrt{E^{\mu}_{1}(x_1)}\sqrt{E^{\nu}_{g}(x_1)}-\sqrt{E^{\nu}_{1}(x_1)}\sqrt{E^{\mu}_{g}(x_1)}\bigg|.
\end{split}
\end{equation*}

Let $D_1 := \sqrt{E^{\mu}_{1}(x_1)}\cdot \sqrt{E^{\nu}_{g}(x_1)}+\sqrt{E^{\nu}_{1}(x_1)}\sqrt{E^{\mu}_{g}(x_1)}$, so that
\begin{equation*}
    \begin{split}
    T_1 &\le \frac{A_1A_2}{\delta} \bigg| \frac{E^{\mu}_{1}(x_1)E^{\nu}_{g}(x_1) - E^{\nu}_{1}(x_1)E^{\mu}_{g}(x_1)}{D_1} \bigg|\le \frac{A_1A_2}{2\delta^2} \bigg| E^{\mu}_{1}(x_1)E^{\nu}_{g}(x_1) - E^{\nu}_{1}(x_1)E^{\mu}_{g}(x_1)\bigg|,
    \end{split}
\end{equation*}
which again follows from \eqref{lowerbound1}. Now, by triangle inequality, we have that:
\begin{equation*}
\begin{split}
&T_1\le \frac{A_1A_2}{2\delta^2} \big| E^{\nu}_{g}(x_1)\big| \big| E^{\mu}_{1}(x_1) - E^{\nu}_{1}(x_1)\big| +\frac{A_1A_2}{2\delta^2} \big| E^{\nu}_{1}(x_1) \big| \big|E^{\nu}_{g}(x_1)- E^{\mu}_{g}(x_1)\bigg|\\
&\le \frac{A_1A_2\big((A_1^2A_3 +\delta)M_1 + (A_3 +\delta)M_1\big)}{2\epsilon\delta^2}\mathcal{W}_{2}(\mu,\nu):= C_1\mathcal{W}_{2}(\mu,\nu)\\
\end{split}
\end{equation*}
for $C_1 =A_1A_2\frac{\big((A_1^2A_3 +\delta)M_1 + (A_3 +\delta)M_1\big)}{2{\epsilon}\delta^2} = O\big(\frac{1}{\epsilon \delta^2}\big).$ \\

\textbf{Second term: }\\
Following similar steps as for the first term and under assumptions \ref{Word:A1} - \ref{Word:A3}, we show that for $C_2$ a constant, the second term of $T_2 \le C_2|x_1 - x_2|$:
\begin{equation*}
\begin{split}
    T_2 
     &\le \frac{A_1A_2(A_3 +\delta)}{2\delta^2}\big|E^{\nu}_{g}(x_2)- E^{\nu}_{g}(x_1)\big| +\frac{A_1A_2(A_1^2A_3 +\delta)}{2\delta^2} \big| E^{\nu}_{1}(x_1)- E^{\nu}_{1}(x_2)\big|\\
     &  \le \frac{A_1A_2(A_3 +\delta)}{2\delta^2}A_1^2 \frac{L_K}{\epsilon}|x_1-x_2| +\frac{A_1A_2(A_1^2A_3 +\delta)}{2\delta^2} \frac{L_K}{\epsilon}|x_1-x_2|,
\end{split}
\end{equation*}
where the last inequality follows from the Lipschitz continuity of $K(\cdot)$ as in \ref{Word:A3}.
We finally get that $ T_2 \le C_2|x_1-x_2|$ for $C_2 = \big( \frac{L_K A_1^3 A_2 (A_3 +\delta)}{2{\epsilon}\delta^2} + \frac{L_K A_1A_2(A_1^2A_3 +\delta)}{2{\epsilon}\delta^2} \big)$,
proving the Lipschitz condition with $C_2 = O\big(\frac{1}{{\epsilon}\delta^2}\big).$ 


\bigbreak
\textbf{Third Term:}\\
By the Lipschitz continuity of $g(\cdot)$ as in \ref{Word:A1}, \ref{Word:A2}, Remark \ref{Word:Remark1}, and Remark \ref{Word:Remark2} we have,
\begin{equation*}
\bigg|\sigma_{\text{Dup}}(t_1,e^{x_1})\frac{\sqrt{E^{\nu}_{1}(x_2)}}{\sqrt{E^{\nu}_{g}(x_2)}}\bigg| | g(y_1)- g(y_2)| 
\le \frac{A_2\sqrt{A_3+\delta}L_g}{\sqrt{\delta}}|y_1-y_2|:=C_3|y_1-y_2|,
\end{equation*}
where $C_3 = \frac{A_2\sqrt{A_3+\delta}L_g}{\sqrt{\delta}}= O\big(\frac{1}{\sqrt{\delta}}\big)$ is a constant.

\bigbreak
\textbf{Fourth Term:}\newline
Similarly, the last bound in \eqref{splitted1} follows directly from \ref{Word:A2} so that together with Remark \ref{Word:Remark1} and Remark  \ref{Word:Remark2} we have,
\begin{multline*}
    \bigg|g(y_2) \frac{\sqrt{E^{\nu}_{1}(x_2)}}{\sqrt{E^{\nu}_{g}(x_2)}} \bigg| |\sigma_{\text{Dup}}(t_1,e^{x_1})- \sigma_{\text{Dup}}(t_2,e^{x_2})|\le\\
    \le \frac{A_1\sqrt{A_3+\delta}L_{\text{Dup}}}{\sqrt{\delta}}\big(|t_1-t_2|^{1/2}+|x_1 - x_2|\big) := C_4 \big(|t_1-t_2|^{1/2}+|x_1 - x_2|\big),
\end{multline*}
where $L_{\text{Dup}}$ is a constant independent of ${\epsilon}$ and $\delta$ and $C_4 = \frac{A_1\sqrt{A_3+\delta}L_{\text{Dup}}}{\sqrt{\delta}} = O\big(\frac{1}{\sqrt{\delta}}\big)$.
\bigbreak
Putting everything together,\newline
$|\tilde{\sigma}_{\epsilon}(t_1,(x_1,y_1) , \mu) - \tilde{\sigma}_{\epsilon}(t_2,(x_2,y_2) , \nu)|\le L_{\sigma}\big(|t_1-t_2|^{1/2}+|x_1-x_2|+|y_1-y_2|+\mathcal{W}_{2}(\mu,\nu)\big)$
where $L_{\sigma} := \text{max}\{C_1,C_2+C_4,C_3\}$ and therefore is a constant dependent only on $\delta$ and ${\epsilon}$.
The Lipschitz continuity of the drift with a constant $L_b = O(\frac{1}{{\epsilon}\delta^2})$ follows by analogous steps.
Condition (i) follows by taking $L:= \text{max}\{L_b,L_{\sigma}\}= O(\frac{1}{{\epsilon}\delta^2})$. The proof of the linear growth condition (ii) is a straightforward application of the Lipschitz regularity of the drift and diffusion coefficients.
This completes the proof of Proposition \ref{condition} with $L = O(\frac{1}{{\epsilon}\delta^2}).$
\end{proof}
\noindent
\setword{\textbf{Assumptions B}}{Word:AB}
\newline 
\setword{\textbf{B1.}}{Word:B1} \ $(X_{\epsilon,0},Y_{\epsilon,0}) \in L^p(\mathcal{F}_0;\mathbb{R}^2;\mathbb{P})$, $p\ge 2$, is 
independent of the Brownian motion.
\newline
\setword{\textbf{B2.}}{Word:B2} \ 
$ \mathbb{E}\Bigg[\bigg(\int^T_0 |{b}_{\epsilon}(t,0,\mu_0)|\diff t \bigg)^2 \Bigg] + 
\mathbb{E}\Bigg[\bigg(\int^T_0 |{\sigma}_{\epsilon}(t,0,\mu_0)|\diff t \bigg)^2 \Bigg] < \infty.$

\begin{theorem}\label{maintheorem} 
Under assumptions \ref{Word:A1}-\ref{Word:A3}, \ref{Word:B1} and \ref{Word:B2}, there exists a unique solution 
$(X_{\epsilon}, Y_{\epsilon}) \in \mathcal{S}^2\big(\big[0,T\big]\big)$ 
to \eqref{sdeapprox}. 
\end{theorem}
\begin{proof}
It is clear that the drift and diffusion coefficients of process $Y$ in \eqref{sdeapprox} are Lipschitz continuous with respect to the state variable, satisfy the linear growth condition and are $\frac{1}{2}$\textendash H\"older in time. Then the result follows from Theorem 3.1 in \cite{well-posedness}, and the Lipschitz-continuity proved in Proposition \ref{condition}.
\end{proof}

\section{Particle method and propagation of chaos}
\label{chap:particlemethod}

To simulate the McKean\textendash Vlasov SDE \eqref{sde} that describes the dynamics of the log process $X$, we approximate the conditional expectation term $\mathbb{E}[g^{2}(Y)|X=x]$ using the particle method as introduced in \cite{book}.
We refer to \cite{bossytalay} for the particle method and a time stepping scheme for generic McKean--Vlasov equations.

Let $(\textbf{X}^N_t)_{t\in[0,T]}:= \big(X_t^{1,N}, X_t^{2,N}, ..., X_t^{N,N}\big)^\intercal_{t\in[0,T]} $ denote the interacting particle system, and $(\textbf{Y}^N_t)_{t\in[0,T]}:= \big(Y_t^{1,N}, Y_t^{2,N}, ..., Y_t^{N,N}\big)^\intercal_{t\in[0,T]} $ independent Monte Carlo samples. 
We follow \cite{book} to use the Nadaraya--Watson estimator 
\begin{equation} \label{NW}
\mathbb{E}[g^{2}(Y)|X=x] \approx \frac{\frac{1}{N}\sum_{i=1}^{N}g^{2}(Y^{i,N})\Phi_{\epsilon}(X^{i,N}-x)}{\frac{1}{N}\sum_{i=1}^{N}\Phi_{\epsilon}(X^{i,N}-x)},
\end{equation}
where $\Phi_{\epsilon}(\cdot)$ is a regularizing kernel function of the form \eqref{kernel}.
Here, the true measure $\mu_t$ of the joint law of $(X_t,Y_t)$ is approximated by  $\mu^{(\textbf{X}^N_t,\textbf{Y}^N_t)}_t$,
where
\begin{eqnarray}
\nonumber
\diff X^{i,N}_t &=& b_N(t,(X_t^{i,N},Y_t^{i,N}), \mu_t^{(\textbf{X}^N_t,\textbf{Y}^N_t)})\diff t+\sigma_N(t,(X_t^{i,N},Y_t^{i,N}),\mu_t^{(\textbf{X}^N_t,\textbf{Y}^N_t)})\diff{W^{x,i}_t},\\
\diff Y^{i,N}_t &=& m(\theta - Y^{i,N}_t)\diff t + \gamma \diff W^{y,i}_t,
\label{simulate}
\end{eqnarray}
\begin{equation*}
\begin{split}
   \text{with }\sigma_N(t,(X_t^{i,N},Y_t^{i,N}),& \mu_t^{(\textbf{X}^N_t,\textbf{Y}^N_t)}) = \\
   =&g(Y^{i,N}_t)\,\sigma_{\text{Dup}}(t,X^{i,N}_t)\frac{\sqrt{\sum_{j=1}^{N}\Phi_{\epsilon}(X_t^{j,N} - X_t^{i,N})}}{\!\!\!\!\sqrt{\sum_{j=1}^{N}g^{2}(Y_t^{j,N})\Phi_{\epsilon}(X_t^{j,N} - X_t^{i,N})}}, 
\end{split}
\end{equation*}
$b_N=-\sigma_N^2/2$,
with  $\diff W^{x,i}_t\diff W^{y,i}_t = \rho_{X,Y} \diff t$ and
with independent $(X_0^{i,N}, Y_0^{i,N})$. 

The interaction term $\mu^{(\textbf{X}^N_t,\textbf{Y}^N_t)}_t$ distinguishes the particle method from the classical Monte Carlo method, since the paths in the former are no longer independent.
The particle method is only useful if it converges to the McKean--Vlasov SDE describing the dynamics of the regularised calibrated LSV model.
We will study \textit{strong propagation of chaos} below.

Using the general assumptions and the regularity of the coefficients proved in 
Proposition \ref{condition}, the following is a direct consequence of \cite{smith}, Proposition 3.1.


\begin{proposition}
\textit{Let $(X^{i,N}_t)$ be the solution to equation \eqref{simulate} and $X^{i}_{\epsilon,t}$ be 
solutions to \eqref{sdeapprox} driven by the respective Brownian motions $(W^{x,i},W^{y,i})$.
Then under assumptions  \ref{Word:A1}-\ref{Word:A3}, \ref{Word:B1} with $p\ge 4$ and \ref{Word:B2},}
\begin{equation*}
\sup_{i \in \{1,..,N\}} \mathbb{E} \Bigg[\sup_{t \in [0,..,T]} \lvert X^{i,N}_t - X^i_{\epsilon,t} \rvert^2\Bigg] \le C N^{-\frac{1}{2}}.
\end{equation*}
\end{proposition}
 \bigbreak

\section{Convergence of an Euler--Maruyama scheme}
\label{chap:euler}
\label{chap:scheme}

To simulate \eqref{simulate}, 
we use the classical Euler--Maruyama scheme with $M$ uniform time-steps of width $\Delta t = T/M$. Specifically, let $\{t_0=0,t_1,t_2, ... ,t_M=T\}$ denote the time discretisation of $[0,T]$ so that $t_m = m\Delta t$ and for $m \in \{0,1,...,M-1\}$,
\begin{equation} \label{EM}
\begin{split}
&X^{i,N,M}_{t_{m+1}} = X^{i,N,M}_{t_{m}}+ b_N(t_m,(X^{i,N,M}_{t_m},Y^{i,N,M}_{t_m}),\mu_t^{(\textbf{X}^N_t,\textbf{Y}^N_t)})\Delta t+\\
&\quad\quad\quad\quad+\sigma_N(t_m,(X^{i,N,M}_{t_m},Y^{i,N,M}_{t_m}),\mu_t^{(\textbf{X}^N_t,\textbf{Y}^N_t)})\Delta W^{x,i}_{t_m}, \quad X^{i,N,M}_{0} = X^i_0 \in\mathbb{R},\\
&Y^{i,N,M}_{t_{m+1}} = Y^{i,N,M}_{t_{m}} + m(\theta - Y^{i,N,M}_{t_{m}})\Delta t+ \gamma \Delta W^{y,i}_{t_m},\quad\quad\quad\quad  Y^{i,N,M}_{0} = Y^i_0 \in \mathbb{R},
\end{split}
\end{equation}
where $\Delta W^{\cdot,i}_{t_m}=W^{\cdot,i}_{t_{m+1}}-W^{\cdot,i}_{t_m}$, that is $\Delta W^{\cdot,i}_{t_m} \sim N(0,\Delta t)$, and increments $\Delta W^{x,i}_{t_m}, \Delta W^{y,i}_{t_m}$ have correlation $\rho_{x,y}$. 

It is well-established (see, e.g., \cite{doksasi}) that for a classical SDE with  Lipschitz-regular drift and diffusion coefficients, the standard explicit Euler--Maruyama scheme converges strongly with order $1/2$ in the step-size.
For particle approximations to McKean--Vlasov equations, the exchangeability and assumed regularity in the measure component allows for error bounds of order 1/2 that are independent of $N$, as shown, e.g., in \cite{smith}.

We now revisit this result and prove the strong convergence of the explicit Euler--Maruyama scheme for the particle system dynamics \eqref{simulate} to find the exact relationship between the rate of convergence and the regularisation parameters $\epsilon$ and $\delta$.
 To establish the results below, we use the Lipschitz regularity of the drift and diffusion coefficients in the state and measure variables that we derived for equation \eqref{sdeapprox}.


We first introduce the continuous-time version of the discretised process defined in \eqref{EM}. Let $m_t := \text{max}\{m\in\{0,...,M-1\}:t_m \le t\} $, $t' := \text{max}\{t_m, m\in\{0,...,M-1\}:t_m \le t\}$, $Z_t := (X_t,Y_t)$, $W_t := (W^x,W^y)$, and $\mu_{t}^{\hat{\textbf{Z}}^{N}}$ denote the law of $\hat{\textbf{Z}}^{N}$. For $t\in[0,T]$, we define the continuous-time process by
\begin{equation}\label{interpolant}
\diff{\hat{Z}^{i,N}_t} =  b_N(t',\tilde{Z}^{i,N}_{t},\mu_{t}^{\tilde{\textbf{Z}}^{N}})\diff t+ \sigma_N(t',\tilde{Z}^{i,N}_{t},\mu_{t}^{\tilde{\textbf{Z}}^{N}})\diff {W_t^i},
\end{equation}
where $\tilde{Z}^{i,N}_{t} := \hat{Z}^{i,N}_{t'}$ is a piecewise constant process, and $\mu_{t}^{\tilde{\textbf{Z}}^{N}}:=\mu_{t'}^{\hat{\textbf{Z}}^{N}}$ is the associated approximation to the true measure.

The proofs of Theorem \ref{strong}, Proposition \ref{onestepestimate}, and Proposition \ref{momentstability} follow the procedure from \cite{adaptive} and \cite{smith}, but keep track of the dependence of all bounds on the Lipschitz constant from Proposition \ref{condition} and hence on the regularisation parameters $\epsilon$ and $\delta$.

 
\begin{proposition}[One-step estimate] \label{onestepestimate} Let $\hat{Z}_t^{i,N}$ be the solution to \eqref{interpolant} and $Z_0 \in L_0^2(\mathbb{R}^2)$. Under assumptions \ref{Word:A1}-\ref{Word:A3}, \ref{Word:B1} and \ref{Word:B2}, there exist positive constants $C_L = O(L^4e^{L^2}), L = O\big(\frac{1}{\epsilon \delta^2}\big)$, such that
\begin{equation}
\mathbb{E}\bigg[\sup_{s \in [0,t]}|\hat{Z}^{i,N}_{s} -\hat{Z}^{i,N}_{s'}|^2\bigg] \le C_L\Delta t.
\end{equation}
\end{proposition}
\begin{proof}
From equation \eqref{interpolant} it is straightforward that 
\begin{equation*}
\begin{split}
&|\hat{Z}^{i,N}_{s} -\hat{Z}^{i,N}_{s'}|^2 = 
|b_N(s',\hat{Z}^{i,N}_{s'},\mu_{s'}^{\hat{\textbf{Z}}^{N}})(s-s') + \sigma_N(s',\hat{Z}^{i,N}_{s'},\mu_{s'}^{\hat{\textbf{Z}}^{N}})(W_s^i-W_{s'}^i)|^2\\
&\le 2|b_N(s',\hat{Z}^{i,N}_{s'},\mu_{s'}^{\hat{\textbf{Z}}^{N}})(s-s')|^2 +2| \sigma_N(s',\hat{Z}^{i,N}_{s'},\mu_{s'}^{\hat{\textbf{Z}}^{N}})(W_s^i-W_{s'}^i)|^2.
\end{split}
\end{equation*}
Hence, applying Chebyshev's integral inequality and It\^o's isometry,
\begin{equation*}
\begin{split}
\mathbb{E}&\bigg[|\hat{Z}^{i,N}_{s} -\hat{Z}^{i,N}_{s'}|^2\bigg] \le  2\mathbb{E}\bigg[\bigg(\int_{s'}^{s} \big|b_N(r',\tilde{Z}^{i,N}_{r},\mu_{r}^{\tilde{\textbf{Z}}^{N}})\big|\diff r\bigg)^2\bigg] +\\
& \quad \quad\quad\quad\quad\quad\quad \quad\quad\quad\quad\quad\quad \quad\quad\quad\quad + 2 \mathbb{E}\bigg[\big|\int_{s'}^{s} \sigma(r',\tilde{Z}^{i,N}_{r},\mu_{r}^{\tilde{\textbf{Z}}^{N}})\diff {W_{r}^i}\big|^2\bigg] \\
&\le 2 (s-s') \mathbb{E}\bigg[\int_{s'}^{s} |b_N(r',\tilde{Z}^{i,N}_{r},\mu_{r}^{\tilde{\textbf{Z}}^{N}})|^2\diff r\bigg]+ 2\mathbb{E}\bigg[\int_{s'}^{s} |\sigma(r',\tilde{Z}^{i,N}_{r},\mu_{r}^{\tilde{\textbf{Z}}^{N}})|^2 \diff r\bigg]
\end{split}
\end{equation*}
By the linear growth of $b_N$ and $\sigma_N$ and the moment stability of $\tilde{Z}^{i,N}_{r}$ in Proposition \ref{momentstability}, we have that for the Lipschitz constant $L>0$ that we obtain in Proposition \eqref{condition},
\begin{equation*}
\begin{split}
&\mathbb{E} \bigg[\sup_{r \in [0,s]}|b_N(r',\tilde{Z}^{i,N}_{r},\mu_{r}^{\tilde{\textbf{Z}}^{N}})|^2 \bigg] \le 2L^2(1+\mathbb{E}\big[\sup_{r \in [0,s]}|\tilde{Z}^{i,N}_{r}|^2\big]) \le C_L \textnormal{, and }\\
&\mathbb{E} \bigg[\sup_{r \in [0,s]}|\sigma_N(r',\tilde{Z}^{i,N}_{r},\mu_{r}^{\tilde{\textbf{Z}}^{N}})|^2 \bigg] \le  2L^2(1+\mathbb{E}\big[\sup_{r \in [0,s]}|\tilde{Z}^{i,N}_{r}|^2\big]) \le C_L,
\end{split}
\end{equation*}
where $C_L= O(L^4e^{L^2})$ is a positive constant. Therefore, for all $t\in [0,T]$,
\begin{equation*}
\begin{split}
\mathbb{E}\bigg[\sup_{s \in [0,t]}|\hat{Z}^{i,N}_{s} -\hat{Z}^{i,N}_{s'}|^2\bigg] \le C_L |s-s'| \le C_L\Delta t.
\end{split}
\end{equation*}
This completes the proof with $C_L= O(L^4e^{L^2}), L=O\big(\frac{1}{\epsilon \delta^2}\big)$.
\end{proof}

\begin{proposition}[Moment stability]\label{momentstability} Let $\hat{Z}^{i,N}_{t}$ be the solution to \eqref{interpolant} and $Z_0 \in L_0^2(\mathbb{R}^2)$. Under assumptions \ref{Word:A1}-\ref{Word:A3}, \ref{Word:B1} and \ref{Word:B2}, there exist positive constants $\tilde{C}= O(L^2e^{L^2}), L=O(\frac{1}{\epsilon \delta^2})$, such that
\begin{equation*}
\max_{i \in \{1,..,N\}} \mathbb{E}\Bigg[\sup_{t \in [0,T]} \lvert \hat{Z}^{i,N}_{t} \rvert^2\Bigg] \le \tilde{C}.
\end{equation*}
\end{proposition}

\begin{proof}
Applying It\^o's lemma to $|\hat{Z}^{i,N}_{t}|^2$ and integrating over time gives that  
\begin{equation*}
\begin{split}
|\hat{Z}^{i,N}_{t}|^2 =& |\hat{Z}^{i,N}_{0}|^2 +  \int_0^t 2\langle |\hat{Z}_s^{i,N}|, b_N(s',\tilde{Z}^{i,N}_{s},\mu_{s}^{\tilde{\textbf{Z}}^{N}})\rangle \diff s +\\
&+\int_0^t 2 \langle|\hat{Z}^{i,N}_{s}|, |\sigma_N(s',\tilde{Z}^{i,N}_{s},\mu_{s}^{\tilde{\textbf{Z}}^{N}})|\diff {W_s^i}\rangle+\int_0^t  |\sigma_N(s',\tilde{Z}^{i,N}_{s},\mu_{s}^{\tilde{\textbf{Z}}^{N}})|^2\diff s,
\end{split}
\end{equation*}
so that $\forall t \in [0,T]$,\newline
\begin{equation}\label{enikserw}
\begin{split}
&\mathbb{E}[\sup_{s \in [0,t]}|\hat{Z}^{i,N}_{s}|^2] \le \mathbb{E}\bigg[|\hat{Z}^{i,N}_{0}|^2\bigg] +2\int_0^t\mathbb{E}\bigg[\sup_{u \in [0,s]} \langle |\hat{Z}_u^{i,N}|, b_N(u',\tilde{Z}^{i,N}_{u},\mu_{u}^{\tilde{\textbf{Z}}^{N}})\rangle\bigg] \diff s\\
&+2\mathbb{E}\bigg[\sup_{s \in [0,t]}\int_0^s \langle|\hat{Z}^{i,N}_{u}|, |\sigma_N(u',\tilde{Z}^{i,N}_{u},\mu_{u}^{\tilde{\textbf{Z}}^{N}})|\diff {W_u^i}\rangle \bigg] +\\
&+\int_0^t \mathbb{E}\bigg[\sup_{u \in [0,s]} |\sigma_N(u',\tilde{Z}^{i,N}_{u},\mu_{u}^{\tilde{\textbf{Z}}^{N}})|^2 \bigg] \diff s.
\end{split}
\end{equation}
By the linear growth and Lipschitz regularity of $b$ as in Proposition \ref{condition},
\begin{equation*}
\begin{split}
\langle &|\hat{Z}_u^{i,N}|,|b_N(u',\tilde{Z}^{i,N}_{u},\mu_{u}^{\tilde{\textbf{Z}}^{N}})|\rangle \le \langle |\hat{Z}_u^{i,N}|,(b_N(u',\tilde{Z}^{i,N}_{u},\mu_{u}^{\tilde{\textbf{Z}}^{N}}) - b_N(u',0,\nu^0_u) \rangle + \\
&+ \langle |\hat{Z}_u^{i,N}|,b_N(u',0,\nu^0_u)\rangle
\le\frac{1}{2}\bigg(2|\hat{Z}_u^{i,N}|^2 + L^2|\tilde{Z}^{i,N}_{u}|^2 + \frac{L^2}{N}\sum_{j=1}^N |\tilde{Z}^{j,N}_{u}|^2 + L^2\bigg),
\end{split}
\end{equation*}
where $\nu^0_u$ denotes the approximation to the true measure corresponding to state 0.
Therefore, for a positive constant $A_L = O(L^2)$,
\begin{equation}
\mathbb{E}\bigg[\sup_{u \in [0,s]}\langle |\hat{Z}_u^{i,N}|,|b_N(u',\tilde{Z}^{i,N}_{u},\mu_{u}^{\tilde{\textbf{Z}}^{N}})|\rangle\bigg] \le A_L \bigg(1+ \mathbb{E}\bigg[\sup_{u \in [0,s]} |\hat{Z}^{i,N}_{u}|^2\bigg]\bigg).
\end{equation}

By similar arguments we get that for $\Tilde{A}=O(L)$ and $\Tilde{B}=O(L^2)$ positive constants,
\begin{equation*}
\begin{split}
\mathbb{E}\bigg[\sup_{u \in [0,s]} |\sigma(u',\tilde{Z}^{i,N}_{u},\mu_{u}^{\tilde{\textbf{Z}}^{N}})|\bigg] &\le \tilde{A}\, \mathbb{E}\bigg[\sup_{u \in [0,s]} \bigg( 1+ |\hat{Z}^{i,N}_{u}| + \bigg(\frac{1}{N}\sum_{j=1}^{N} |\hat{Z}^{j,N}_{u}|^2 \bigg)^{1/2}\bigg) \bigg],
\end{split}
\end{equation*}
\begin{equation*}
\begin{split}
\mathbb{E}\bigg[\sup_{u \in [0,s]} |\sigma(u',\tilde{Z}^{i,N}_{u},\mu_{u}^{\tilde{\textbf{Z}}^{N}})|^2\bigg] &\le \tilde{B}\, \mathbb{E}\bigg[\sup_{u \in [0,s]} \bigg( 1+ |\hat{Z}^{i,N}_{u}|^2 + \frac{1}{N}\sum_{j=1}^{N} |\hat{Z}^{j,N}_{u}|^2\bigg) \bigg].
\end{split}
\end{equation*}
Returning to \eqref{enikserw}, we apply the Burkholder--Davis--Gundy inequality, to get that 
\begin{equation*}
\begin{split}
2\mathbb{E}\bigg[\sup_{s \in [0,t]}\int_0^s &\langle|\hat{Z}^{i,N}_{u}|, |\sigma(u',\tilde{Z}^{i,N}_{u},\mu_{u}^{\tilde{Z}^{N}})|\diff {W_u^i}\rangle \bigg] \le \\
&\le C_L \mathbb{E}\Bigg[ \int_0^t |\hat{Z}^{i,N}_{s}|\cdot \bigg( 1+ |\hat{Z}^{i,N}_{s}| +\bigg( \frac{1}{N}\sum_{j=1}^{N} |\hat{Z}^{j,N}_{s}|^2\bigg)^{1/2} \bigg) \diff s \Bigg]\\
&\le C_L \mathbb{E}\Bigg[ \int_0^t \frac{1}{2}|\hat{Z}^{i,N}_{s}|^2 +\frac{1}{2} \bigg( 1+ |\hat{Z}^{i,N}_{s}| +\bigg( \frac{1}{N}\sum_{j=1}^{N} |\hat{Z}^{j,N}_{s}|^2\bigg)^{1/2} \bigg)^2 \diff s \Bigg],\\ 
\end{split}
\end{equation*}
which follows from Young's inequality and where $C_L = O(L)$ is a  positive constant. By the linearity of the expectation and since the processes $\hat{Z}^{j,N}_{s}$ are identically distributed, one concludes that for a positive constant $\tilde{C}_L =O(L)$,
\begin{equation*}
\begin{split}
\max_{i \in \{1,..,N\}} 2\,\mathbb{E}\bigg[\sup_{s \in [0,t]}\int_0^s \langle|\hat{Z}^{i,N}_{u}|, |\sigma(u',&\tilde{Z}^{i,N}_{u},\mu_{u}^{\tilde{Z}^{N}})|\diff {W_u^i}\rangle \bigg] \le\\
&\le \tilde{C}_L\, \mathbb{E}\Bigg[ \int_0^t \big(1+ \max_{i \in \{1,..,N\}} \mathbb{E}\bigg[ |\hat{Z}^{i,N}_{s}|^2 \bigg]\big) \diff s \Bigg].
\end{split}
\end{equation*}
Taking the maximum over the index $i$ in equation \eqref{enikserw} and equipped with the above bounds, we have that for a positive constant $\Tilde{C} = O(L^2)$, 
\begin{equation*}
\max_{i \in \{1,..,N\}} \mathbb{E}[\sup_{s \in [0,t]}|\hat{Z}^{i,N}_{s}|^2] \le \tilde{C}\, \mathbb{E}\Bigg[ \int_0^t \big(1+ \max_{i \in \{1,..,N\}} \mathbb{E}\bigg[\sup_{u \in [0,s]} |\hat{Z}^{i,N}_{u}|^2\bigg] \big) \diff s \Bigg].
\end{equation*}
Finally, applying Gr\"onwall's inequality we get $\max_{i \in \{1,..,N\}} \mathbb{E}\bigg[\sup_{s \in [0,t]}|\hat{Z}^{i,N}_{s}|^2\bigg] \le \Tilde{C},$ 
where $\Tilde{C} = O(L^2e^{L^2})$ is a positive constant and $L=O\big(\frac{1}{\epsilon \delta^2}\big)$.
\end{proof}

\begin{theorem}[Strong convergence of Euler--Maruyama scheme]\label{strong}
Let $Z^{i,N}= (X^{i,N}, Y^{i,N}) $ be the solution to \eqref{simulate} and $\hat{Z}^{i,N}$ the solution to \eqref{interpolant}. Also, let $Z_0 \in L_0^2(\mathbb{R}^2)$. Under assumptions \ref{Word:A1}-\ref{Word:A3}, \ref{Word:B1} and \ref{Word:B2}, there exists positive constants $C = O(L^6 e^{2L^2})$, $L = O\big(\frac{1}{\epsilon \delta^2}\big)$ such that
$$\max_{i \in \{1,..,N\}} \mathbb{E}\Bigg[\sup_{t \in [0,T]} \lvert \hat{Z}^{i,N}_t - Z^{i,N}_t \rvert^2\Bigg] \le C\Delta t.$$
\end{theorem}
\begin{proof}
Let $E_t^i:=\hat{Z}^{i,N}_t-Z^{i,N}_t$ so that it satisfies the SDE:
\begin{equation*}
\begin{split}
    &\diff E_t^i =(b_N(t',\tilde{Z}^{i,N}_{t},\mu_{t}^{\tilde{\textbf{Z}}^{N}})- b_N(t,Z^{i,N}_t,\mu_{t}^{{\textbf{Z}}^{N}}))\diff t +\\
    & \quad \quad\quad\quad\quad\quad\quad\quad\quad\quad\quad\quad+(\sigma_N(t',\tilde{Z}^{i,N}_{t},\mu_{t}^{\tilde{\textbf{Z}}^{N}})- \sigma_N(t,Z^{i,N}_t,\mu_t^{\textbf{Z}^N}))\diff W^i_t.
\end{split}
\end{equation*}
By It\^{o}'s lemma we have that
\begin{equation}\label{E^i}
\begin{split}
|E_t^i|^2 &=2 \int_0^t\langle E_s^i, (b_N(s',\tilde{Z}^{i,N}_{s},\mu_{s}^{\tilde{\textbf{Z}}^{N}})- b_N(s,Z^{i,N}_s,\mu_{s}^{{\textbf{Z}}^{N}}))\rangle \diff s \\
&+ 2 \int_0^t\langle E_s^i, \bigg(\sigma_N(s',\tilde{Z}^{i,N}_{s},\mu_{s}^{\tilde{\textbf{Z}}^{N}})- \sigma_N(s,Z^{i,N}_s,\mu_s^{\textbf{Z}^N})\bigg)\diff W^i_s \rangle\\
&+ \int_0^t |\sigma_N(s',\tilde{Z}^{i,N}_{s},\mu_{s}^{\tilde{\textbf{Z}}^{N}})- \sigma_N(s,Z^{i,N}_s,\mu_s^{\textbf{Z}^N})|^2 \diff s.
\end{split}
\end{equation}
Using the inequality $(a+b)^2 \le 2a^2 +2b^2$
and Proposition \ref{condition},
we have that
{\allowdisplaybreaks
\begin{eqnarray*}
\label{sigmaa}
&|\sigma_N(s',\tilde{Z}^{i,N}_{s},\mu_{s}^{\tilde{\textbf{Z}}^{N}})- \sigma_N(s,Z^{i,N}_s,\mu_s^{\textbf{Z}^N})|^2 \le 2|\sigma(s',\tilde{Z}^{i,N}_{s},\mu_{s}^{\tilde{\textbf{Z}}^{N}})-\\
&-\sigma_N(s,\tilde{Z}^{i,N}_{s},\mu_{s}^{\tilde{\textbf{Z}}^{N}})|^2 +2|\sigma_N(s,\tilde{Z}^{i,N}_{s},\mu_{s}^{\tilde{\textbf{Z}}^{N}})-\sigma_N(s,Z^{i,N}_s,\mu_s^{\textbf{Z}^N})|^2\\
& \le 2L^2|s'-s| + 2L^2\big(|\tilde{Z}^{i,N}_{s} - Z^{i,N}_s| + \mathcal{W}_{2}(\mu_{s}^{\tilde{\textbf{Z}}^{N}},\mu_s^{\textbf{Z}^N})\big)^2\\
&\le  2L^2\Delta t + 8L^2\big(|\tilde{Z}^{i,N}_{s} - \hat{Z}^{i,N}_s|^2+ |E^i_s|^2 +\frac{1}{N}\sum_{j=1}^{N} |\tilde{Z}^{j,N}_{s} - \hat{Z}^{j,N}_s|^2+ \frac{1}{N}\sum_{j=1}^{N}|E^j_s|^2\big),
\end{eqnarray*}
}
by the triangle inequality for $\mathcal{W}_{2}(\mu,\nu)$ (see, e.g., \cite{Villani})
%
and its standard bound
$ \mathcal{W}_{2}(\mu_{s}^{\tilde{\textbf{Z}}^{N}},\mu_s^{\hat{\textbf{Z}}^{N}}) \le \big(\frac{1}{N}\sum_{j=1}^{N} |\tilde{Z}^{j,N}_{s} - \hat{Z}^{j,N}_s|^2 \big)^{1/2}$. 
From Proposition \ref{onestepestimate}, for some $A>0$,
\begin{equation}\label{sigmaa2}
\begin{split}
\mathbb{E}\Bigg[\sup_{s \in [0,t]}|\sigma_N&(s',\tilde{Z}^{i,N}_{s},\mu_{s}^{\tilde{\textbf{Z}}^{N}})- \sigma_N(s,Z^{i,N}_s,\mu_s^{\textbf{Z}^N})|^2\Bigg] \le\\
&\le AL^2\Delta t + 8L^2 \mathbb{E}\big[\sup_{s \in [0,t]}|E^i_s|^2\big]+8L^2\mathbb{E}\big[\sup_{s \in [0,t]}\frac{1}{N}
\sum_{j=1}^{N}|E^j_s|^2\big].
\end{split}
\end{equation}
Returning to \eqref{E^i}, by the BDG inequality, for some $C_L>0, C_L = O(L^2)$, 
\begin{equation*}
\begin{split}
&\mathbb{E}\Bigg[\sup_{s \in [0,t]} \bigg| \int_0^s\langle E_u^i, (\sigma_N(u',\tilde{Z}^{i,N}_{u},\mu_{u}^{\tilde{\textbf{Z}}^{N}})- \sigma_N(u,Z^{i,N}_u,\mu_u^{\textbf{Z}^N}))\diff W^i_u \rangle \bigg| \Bigg] \le \\
& \le C_L \mathbb{E}\Bigg[ \int_0^t E_s^i\cdot \big( |s'-s|^{1/2} + |\tilde{Z}^{i,N}_{s} - {Z}^{i,N}_s| + \mathcal{W}_{2}(\mu_{s}^{\tilde{\textbf{Z}}^{N}},\mu_s^{{\textbf{Z}}^{N}})\big) \diff s \Bigg]\\
&\le C_L \mathbb{E}\Bigg[ \int_0^t\Bigg( \frac{1}{2}|E_s^i|^{2} + \frac{1}{2}\bigg( |s'-s|^{1/2} + |\tilde{Z}^{i,N}_{s} - {Z}^{i,N}_s| + \mathcal{W}_{2}(\mu_{s}^{\tilde{\textbf{Z}}^{N}},\mu_s^{{\textbf{Z}}^{N}})\bigg)^{2}\Bigg) \diff s \Bigg]\\
& \le C_L \mathbb{E}\Bigg[ \int_0^t\Bigg( \frac{1}{2}|E_s^i|^{2} + \frac{3}{2} |s'-s| +\frac{3}{2}|\tilde{Z}^{i,N}_{s} - {Z}^{i,N}_s|^2 +\frac{3}{2} \mathcal{W}_{2}(\mu_{s}^{\tilde{\textbf{Z}}^{N}},\mu_s^{{\textbf{Z}}^{N}})^{2} \Bigg) \diff s \Bigg]\\
&\le C_L\mathbb{E}\Bigg[ \int_0^t \Bigg( \frac{1}{2}|E_s^i|^{2} + \frac{3}{2}\Delta t +3|\tilde{Z}^{i,N}_{s} - \hat{Z}^{i,N}_s|^2 + 3|\hat{Z}^{i,N}_{s} - {Z}^{i,N}_s|^2+\\
& +\frac{3}{2} \mathcal{W}_{2}(\mu_{s}^{\tilde{\textbf{Z}}^{N}},\mu_s^{{\textbf{Z}}^{N}})^{2} \Bigg)\diff s \Bigg] \le  C_L \mathbb{E}\Bigg[  \frac{3}{2}t\Delta t + \int_0^t \Bigg( \frac{7}{2}|E_s^i|^{2} +3|\tilde{Z}^{i,N}_{s} -\hat{Z}^{i,N}_s|^2 +\\
&+\frac{3}{N}\sum_{j=1}^{N} |\tilde{Z}^{j,N}_{s} - \hat{Z}^{j,N}_s|^2 +\frac{3}{N}\sum_{j=1}^{N}|E_s^j|^{2} \Bigg) \diff s \Bigg].
\end{split}
\end{equation*}
Since processes $Z^j$ are identically distributed
and  $\mathbb{E}\big[|\tilde{Z}^{j,N}_{s} - \hat{Z}^{j,N}_s|^2\big] = O(\Delta t)$, by the linearity of expectation, there exists $C=O(L^2)$ a positive constant such that 
\begin{equation*}
\begin{split}
\max_{i \in \{1,..,N\}} \mathbb{E}\bigg[\sup_{s \in [0,t]} \lvert &\int_0^s\langle E_u^i, (\sigma(u',\tilde{Z}^{i,N}_{u},\mu_{u}^{\tilde{\textbf{Z}}^{N}})- \sigma_N(u,Z^{i,N}_u,\mu_u^{\textbf{Z}^N}))\diff W^i_u \rangle \rvert \bigg] \\
&\le Ct\Delta t + C\max_{i \in \{1,..,N\}}\mathbb{E} \bigg[\int_0^t (|E_s^i|^{2})\diff s \bigg].
\end{split}
\end{equation*}

We now consider the first term of equation \eqref{E^i}. For the Lipschitz constant $L$,
\begin{equation*}
\begin{split}
&\langle E_s^i, (b_N(s',\tilde{Z}^{i,N}_{s},\mu_{u}^{\tilde{\textbf{Z}}^N})- b_N(s,Z^{i,N}_s,\mu_{u}^{{\textbf{Z}}^N}))\rangle \le  L E_s^i| \tilde{Z}^{i,N}_s-\hat{Z}^{i,N}_s|+\\
&+L E_s^i| \hat{Z}^{i,N}_s-{Z}^{i,N}_s| + L E_s^i |s'-s|^{1/2} +LE_s^i \mathcal{W}_{2}(\mu_{s}^{\tilde{\textbf{Z}}^{N}},\mu_s^{\hat{\textbf{Z}}^{N}}) + LE_s^i \mathcal{W}_{2}(\mu_{s}^{\hat{\textbf{Z}}^{N}},\mu_s^{{\textbf{Z}}^{N}})\\
&\le \frac{L}{2} \bigg( 6|E_s^i|^2 + | \tilde{Z}^{i,N}_s-\hat{Z}^{i,N}_s|^2+|s'-s|+\frac{1}{N}\sum_{j=1}^N|\tilde{Z}^{i,N}_s-\hat{Z}^{i,N}_s|^2 + \frac{1}{N}\sum_{j=1}^N|E_s^j|^2 \bigg),
\end{split}
\end{equation*}
which follows by $|s'-s| \le \Delta t.$ Now using the one-step estimate of  $|\tilde{Z}^{i,N}_{s} -\hat{Z}^{i,N}_s|^2$ proved in Proposition \ref{onestepestimate}, we have that for constants $\tilde{C_1}, \tilde{C_2}>0$,
\begin{equation*}
\begin{split}
\max_{i \in \{1,..,N\}} \mathbb{E}\bigg[\sup_{s \in [0,t]} 2\int_0^t\langle E_s^i,& (b_N(s',\tilde{Z}^{i,N}_{s})- b_N(s,Z^{i,N}_s))\rangle \diff s \bigg] \\
&\le \tilde{C_1}t\Delta t + \tilde{C_2}\max_{i \in \{1,..,N\}} \mathbb{E}\bigg[\sup_{s \in [0,t]} \int_0^t |E_s^i|^2\diff s\bigg].
\end{split}
\end{equation*}
Substituting the above bounds back into equation \eqref{E^i},
\begin{equation*}
\begin{split}
&\max_{i \in \{1,..,N\}} \mathbb{E}\bigg[\sup_{s \in [0,t]}|E_s^i|^2\bigg] \le \tilde{C_1}t\Delta t + \tilde{C_2}\max_{i} \mathbb{E}\bigg[ \int_0^t\sup_{u \in [0,s]} |E_u^i|^2 \diff s\bigg]+ Ct \Delta t +\\
&+C\max_{i}\mathbb{E} \Bigg[  \int_0^t\sup_{u \in [0,s]} |E_u^i|^{2}\diff s \Bigg]+ \int_0^t \Bigg( AL^2\Delta t + 16L^2\max_{i} \mathbb{E}\bigg[\sup_{u \in [0,s]}|E^i_u|^2\bigg] \Bigg) \diff s\\
&\le \tilde{K}_1t\Delta t +\tilde{K}_2  \int_0^t \max_{i \in \{1,..,N\}} \mathbb{E} \Bigg[\sup_{u \in [0,s]}|E_u^i|^{2} \Bigg]\diff s,
\end{split}
\end{equation*}
where $\tilde{K}_1$ and $\tilde{K}_2$ are positive constants of order $L^6e^{L^2}$ and $L^2$ respectively. By Gr\"onwall's inequality, with $\tilde{K}_2$ non-negative and $ \tilde{K}_1t\Delta t$ non-decreasing, for $t \in [0,T],$ 
\begin{equation}
 \max_{i \in \{1,..,N\}} \mathbb{E}\bigg[\sup_{s \in [0,t]} |E_t^i|^2 \bigg] \le \tilde{K}_1t\Delta t e^{\int_0^t\tilde{K_2}ds} \le C\Delta t,
\end{equation}
where $C = O(L^6 e^{2L^2})$ with $L= O\big(\frac{1}{{\epsilon}\delta^2}\big)$ as in Proposition \ref{condition}.
\end{proof}

\begin{remark}
The dependence of $C = O(L^6 e^{2L^2})$  on $\epsilon$ and $\delta$ (through $L = O\big(\frac{1}{\epsilon \delta^2}\big)$) as predicted by Theorem \ref{strong} will be found pessimistic in our numerical tests.
As regards $\epsilon$, this is because we made no assumptions on the regularity of the distribution of $(X_t,Y_t)$; for a smooth density, it seems plausible that the estimator for conditional expectations is better behaved for small $\epsilon$. A positive $\delta$ was included to prevent a singularity if the denominator in the coefficient approaches 0, but this is not necessary if either there is a positive density in the region where the formula is used (in practice, extrapolation is used outside a compact set), or $g$ has a positive lower bound. 
\end{remark}

\section{Implementation and numerical results}\label{chap:python}
We now consider a Heston-type local volatility model to numerically test the particle method and investigate how the regularisation parameter affects the calibration. The risk-neutral dynamics of a Heston-type local volatility model are
\begin{equation}\label{HestonLVM}
    \begin{split}
    &\diff S_t = r\,S_t \diff t + \sqrt{V_t}\,S_t\, \alpha(t,S_t) \diff W^s_t,\\
    &\diff V_t = k(\theta-V_t)\diff t + \xi\sqrt{V_t}\diff W^v_t,
    \end{split}
\end{equation}
with $\diff W^s_t\diff W^v_t = \rho_{S,V}\diff t$.

We note that this model differs from the one studied in the previous sections but we expect a similar behaviour for this setting as well. Specifically, the CIR volatility process makes the diffusion coefficient in the second component only $\frac{1}{2}$\textendash H\"older in space and not globally Lipschitz continuous. The existing literature on the well-posedness of McKean\textendash Vlasov SDEs (MVSDEs) with non-Lipschitz coefficients and propagation of chaos results is somewhat scarce. An initial work on the strong convergence of the Euler scheme (without order) for MVSDEs assuming only continuity of the coefficients and non-degenerate diffusion (unlike in our case), is \cite{X.Zhang}. Bao and Huang in \cite{BaoHuang} provide results on propagation of chaos and strong convergence of the Euler\textendash Maruyama (EM) scheme for MVSDEs for two cases of H\"older continuous (i) diffusion and (ii) drift coefficients. Another result on the strong convergence of the EM scheme comes from Liu et al. in \cite{Liu.Shi.Wu} for the case of super-linear drift and H\"older continuous diffusion coefficients. Both of these results, however, consider diffusion coefficients that do not depend on the law of the process and therefore cannot be applied directly in our case.

The procedure to calibrate is two-fold. Firstly, having a set of call option prices observed in the market, we calibrate a pure Heston process to get the parameters that best match the market prices according to a chosen optimization technique. Secondly, in each time-step of our discretisation, we calibrate the leverage function $\alpha(t,S)$. Recall the condition for exact calibration, as given in \cite{iff} and adapted for the above Heston-type local volatility model \eqref{HestonLVM}, is 
$\alpha^{2}(t, S)=
\sigma_{\text{Dup}}^{2}(t,S)/\mathbb{E}^{\mathbb{Q}}[V_{t}|S_{t} = S].$
This requires a priori knowledge of the local volatility surface and since there is no knowledge of the option prices for all possible strikes and maturities, then it is necessary to interpolate and extrapolate the local volatility. The authors in \cite{book} propose cubic spline interpolation and flat extrapolation. To approximate the conditional expectation in the leverage function, we use the particle method as in \cite{book}, and revised in Section \ref{chap:particlemethod} above. 
We use the Euler\textendash Maruyama scheme as presented in Section \ref{chap:euler} and repeated below for the Heston-type LSV model \eqref{HestonLVM}.

Finally, having calibrated the model, we are able to estimate European option prices by the average discounted payoff $\frac{1}{N}\sum^N_{i=1} e^{-r\,T}(S^{i,N}_T-K)^{+},$ where $r$ denotes the interest rate and $K$ the strike price of the option.

\bigskip

\textit{Calibration}.
We fix $\delta$ and the bandwidth $\epsilon$.
We set the time-discretisation $\{t_m: m = 0,..,M\} = \{t_0=0,t_1,...,t_M=T\} $ of $[0,T] $ with uniform time-steps of length $\Delta t = T/M$ so that $t_m = m\cdot\Delta t$.
Moreover, below, we use $\Delta W^i_{t_m}=W^i_{t_{m+1}}-W^i_{t_m}$ so that  $\Delta W^i_{t_m} \sim \mathcal{N}(0,\Delta t)$ and $\Delta W^{v,i}_{t_m} = \rho\Delta W^{s,i}_{t_m} + \sqrt{1- \rho^2}Z^i_{t_m}$, where $Z^i_{t_m}$ are independent Brownian motions.
We then follow the following algorithm:\label{alg:one}
\begin{algorithmic}[1]
\State $S^{i,N}_0,V^{i,N}_0 \gets s_0,v_0$\textnormal{ for all i}
\State $\alpha^{2}(0, S^{i,N}_0) \gets \sigma_{\text{Dup}}^{2}(0,S^{i,N}_0)
\frac{\sqrt{N+\delta}}{\sqrt{NV^{i,N}_0+\delta}}$
\While{$m \in \{0,...,M-1\}$} \textnormal{ for all i}
    \State $S^{i,N}_{t_{m+1}} \gets S^{i,N}_{t_{m}} + rS^{i,N}_{t_{m}} \Delta t + \sqrt{V^{i,N}_{t_m}}\,S^{i,N}_{t_m}\, \alpha(t_m,S^{i,N}_{t_m})\Delta W^{s,i}_{t_m}$\;
     \State $V^{i,N}_{t_{m+1}} \gets V^{i,N}_{t_{m}} +k(\theta-V^{i,N}_{t_m})\Delta t + \xi \sqrt{V^{i,N}_{t_m}} \Delta W^{v,i}_{t_m}$\;
      \State $\alpha(t_{m+1},S^{i,N}_{t_{m+1}}) \gets \sigma_{\text{Dup}}(t_m,S^{i,N}_{t_m})\frac{\sqrt{\sum_{j=1}^{N}\Phi_{\epsilon}(S_{t_m}^{j,N}-S^{i,N}_{t_m})+\delta}}{\sqrt{\sum_{j=1}^{N}V_{t_m}^{j,N} \Phi_{\epsilon}(S_{t_m}^{j,N}-S^{i,N}_{t_m})+\delta}}$\;
      \State $m \gets m+1$
\EndWhile
\end{algorithmic}
\bigbreak

Our implementation uses QuantLib, an open-source library for quantitative finance. 
For testing purposes, instead of using real-market call option prices, we generate a volatility surface using the Heston model with  parameters $v_0 = 0.0094$, $\kappa = 1.4124 $, $\theta=0.0137 $, $\xi=0.2988 $, $\rho=-0.1194$, which were calibrated to an FX market in \cite{hestoncal}, and treat this as the market implied surface.
We then alter the initial parameters to $v_0 = 0.014,  \kappa= 1.4, \theta = 0.01, \xi =  0.3, \rho =-0.2$.

In Figure \ref{fig:surfaces}, we plot the artificial ``market" implied volatility surface and the one coming from the pure Heston model with the above modified parameters. We then calibrate the Heston-type local volatility model \eqref{HestonLVM} using the particle method and expect that the leverage function will ``correct'' the difference in the surfaces.

\begin{figure}[H]
   \begin{minipage}{0.5\textwidth}
     \centering
     \includegraphics[width=.8\linewidth]{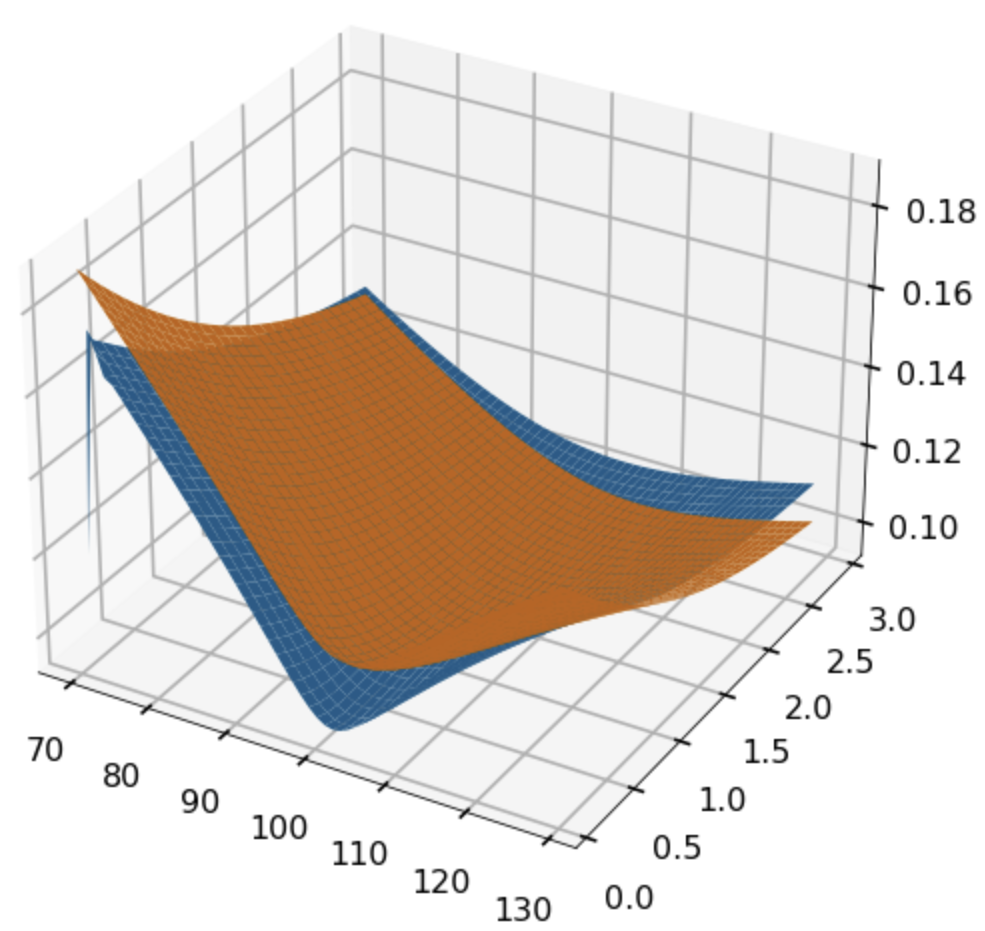}
     \caption{Artificial ``market" (in blue) and pure Heston (in orange) implied volatility surfaces.}\label{fig:surfaces}
   \end{minipage}\hfill
   \begin{minipage}{0.5\textwidth}
     \centering
     \includegraphics[width=.8\linewidth]{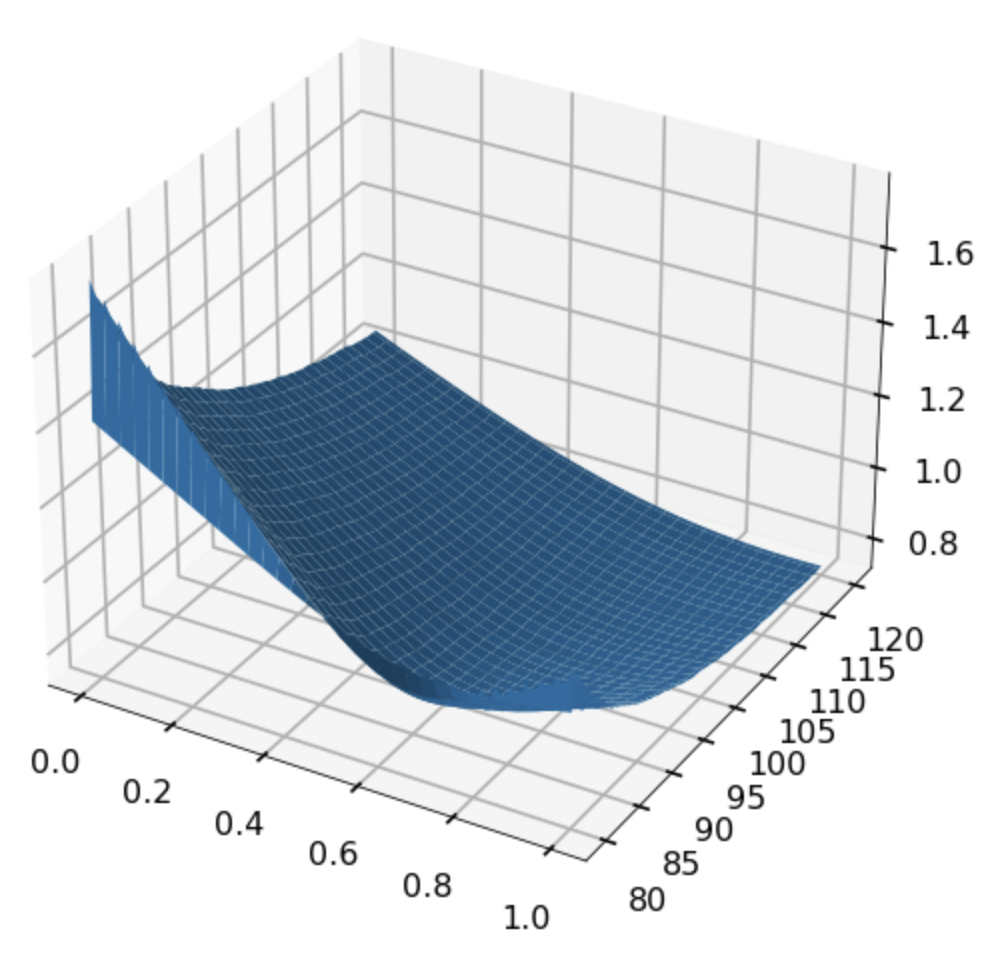}
     \caption{Leverage function for T=$1$Y.}\label{fig:leveragefunction}
   \end{minipage}
\end{figure}

To illustrate our results, in Figure \ref{fig:leveragefunction} we plot the leverage function for European option prices for a range of maturities from $T=0$ to $1Y$ and strikes ranging from $80 \text{ to }120$. We also fix $\epsilon_1 = S_0\,N^{-1/5}$, where $S_0$ is the initial value of process $(S_t)_{t\in [0,T]}$, and $\delta =0.00001$. We choose the bandwidth parameter $\epsilon$ according to the asymptotic mean integrated squared error (AMISE) optimality criterion. It is well-established, see \cite{AMISE}, that the optimal $\epsilon$ that minimises the AMISE of the Nadaraya\textendash Watson kernel density estimator is $c\, N^{-1/5}$, for $c$ a constant. 


We now investigate how the choice of the bandwidth $\epsilon$ and parameter $\delta$  affect the convergence and accuracy of the particle method. To do so, we first compute the Root Mean Square Error (RMSE) as a measure for the difference between the artificial ``market" prices and the prices coming from the calibrated LSV model for European call options with $T =1Y$ maturity for strikes ranging between 80 and 120, for different values of the regularisation parameters.\newline
To price the European call options we follow the above calibration procedure using $M = 100$ time-steps and to save computational time, only $N=10^3$ particles. The running time for the calibration is then at around $7.0$ seconds. We acknowledge that this is only a small number of particles and certain acceleration techniques, as discussed in \cite{book}, could improve the performance of our computations.\newline

Firstly, we fix $\delta = 0.00001$ and alter $\epsilon$. As shown in Table \ref{table:4} and Figures \ref{Fig:h1} and \ref{Fig:h2}, more accurate pricing occurs as $\epsilon$ gets smaller, which is promising in terms of the convergence of the approximation as $\epsilon \to 0.$
On the other hand, our results do not agree with the initial choice of bandwidth by the AMISE criterion.
\begin{table}[H]
\centering
\begin{tabular}{||c c c c c ||} 
 \hline
 $\delta$ fixed.  & $\epsilon_1 =S_0N^{-1/5} $ & $\epsilon_2 = \epsilon_1/10$ & $\epsilon_3 =\epsilon_1/100 $ & $\epsilon_4 = 10\cdot \epsilon_1 $ \\ [0.5ex] 
 \hline\hline
 RMSE &  0.302 &  0.231 &0.068 &  1.34  \\ [1ex] 
 \hline
\end{tabular}
\caption{RMSE for fixed $\delta$ and varying $\epsilon$ .}
\label{table:4}
\end{table}
\begin{figure}[H]
   \begin{minipage}{0.5\textwidth}
     \centering
     \includegraphics[width=.8\linewidth]{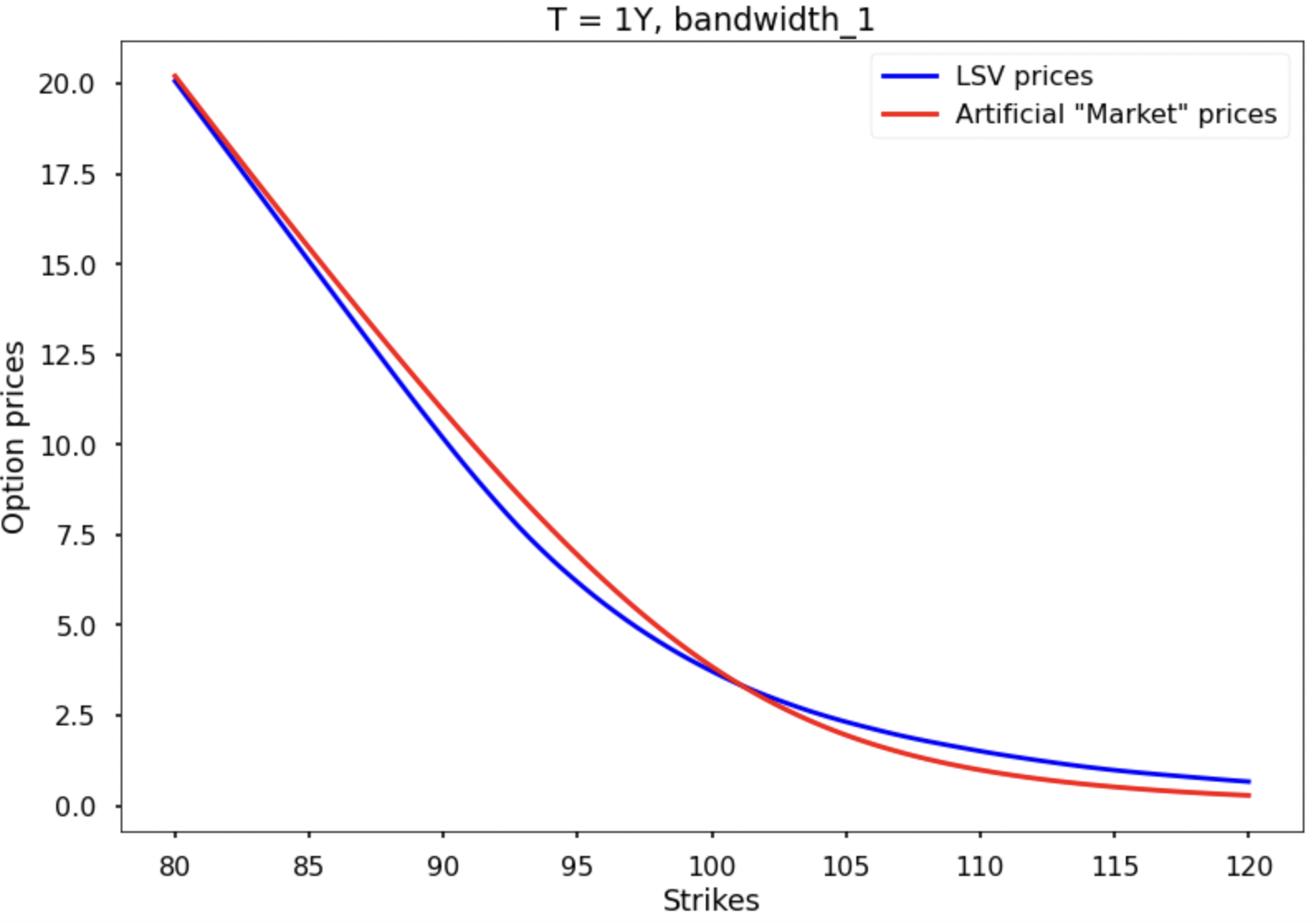}
     \caption{Prices comparison, $\epsilon_1 = S_0N^{-1/5} $}\label{Fig:h1}
   \end{minipage}\hfill
   \begin{minipage}{0.5\textwidth}
     \centering
     \includegraphics[width=.8\linewidth]{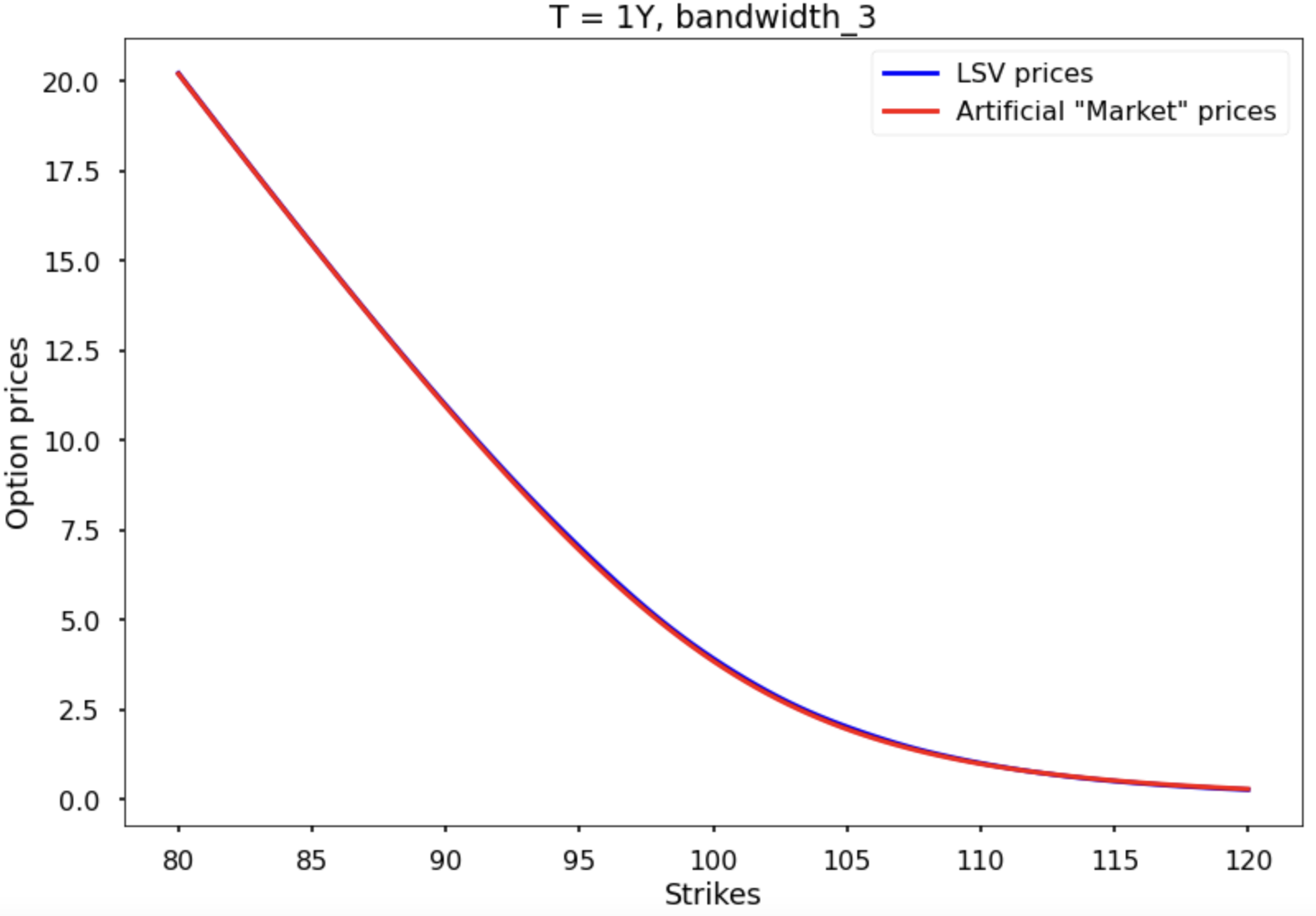}
     \caption{Prices comparison, $\epsilon_3 = \epsilon_1/100$}\label{Fig:h2}
   \end{minipage}
\end{figure}

We now fix $\epsilon = S_0N^{-1/5}/100$, which is the bandwidth that gave the most accurate result above, and alter $\delta$. As shown in Table \ref{table:3}, 
Figures \ref{Fig:d1} and \ref{Fig:d2}, the calibration becomes more accurate as $\delta$ gets smaller, which verifies the convergence of the regularisation as $\delta \to 0.$
\begin{table}[H]
\centering
\begin{tabular}{||c c c c c ||} 
 \hline
 $\epsilon$ fixed.  & $\delta_1 = 0.01 $ & $\delta_2 = 0.001$ & $\delta_3 = 0.0001$ & $\delta_4 = 0.00001$ \\ [0.5ex] 
 \hline\hline
 RMSE &  0.287&  0.196 &0.069 &  0.068  \\ [1ex] 
 \hline
\end{tabular}
\caption{RMSE for fixed $\epsilon$ and varying $\delta$.}
\label{table:3}
\end{table}
\begin{figure}[H]
   \begin{minipage}{0.5\textwidth}
     \centering
     \includegraphics[width=.8\linewidth]{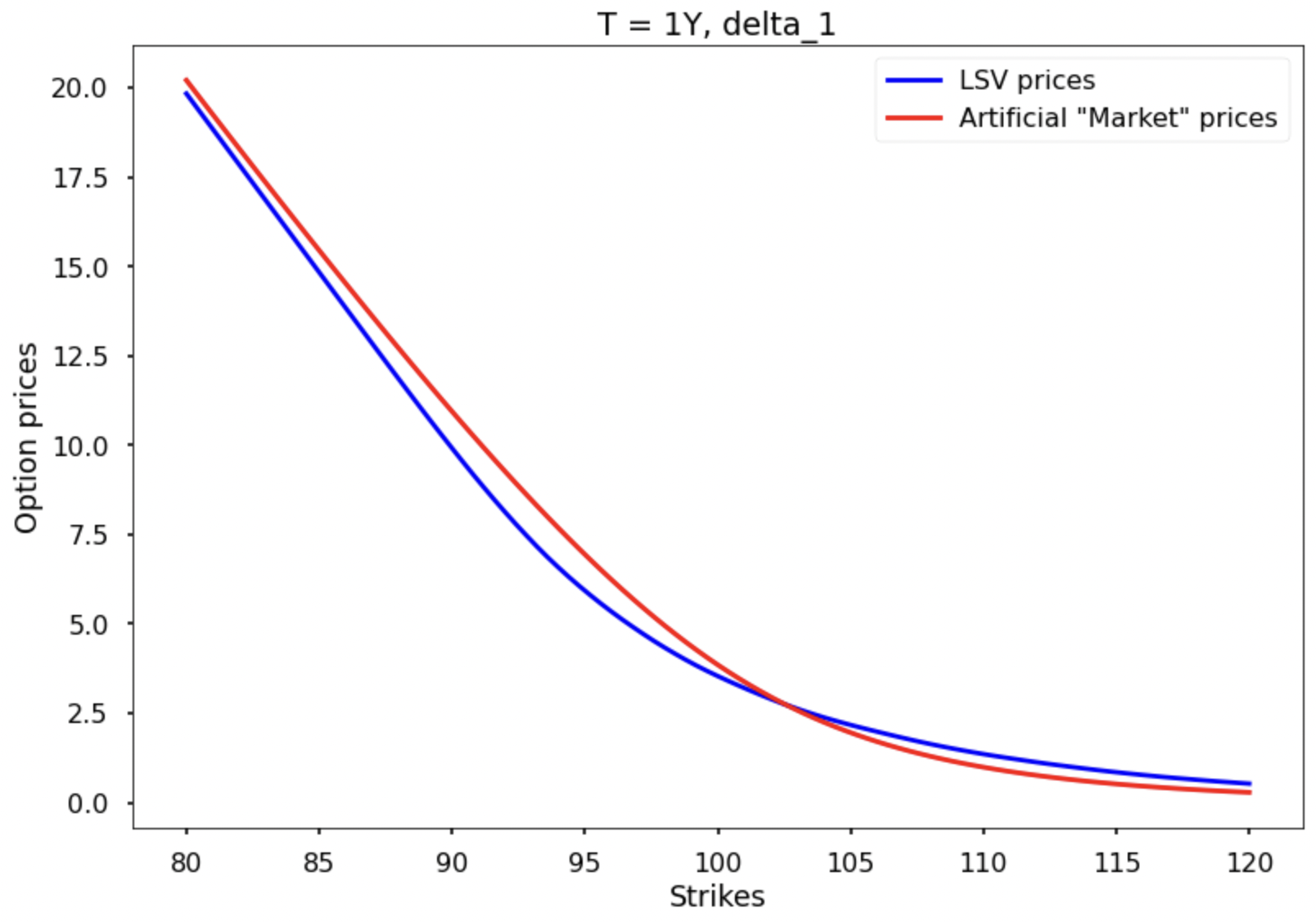}
     \caption{Prices comparison, $\delta_1$}\label{Fig:d1}
   \end{minipage}\hfill
   \begin{minipage}{0.5\textwidth}
     \centering
     \includegraphics[width=.8\linewidth]{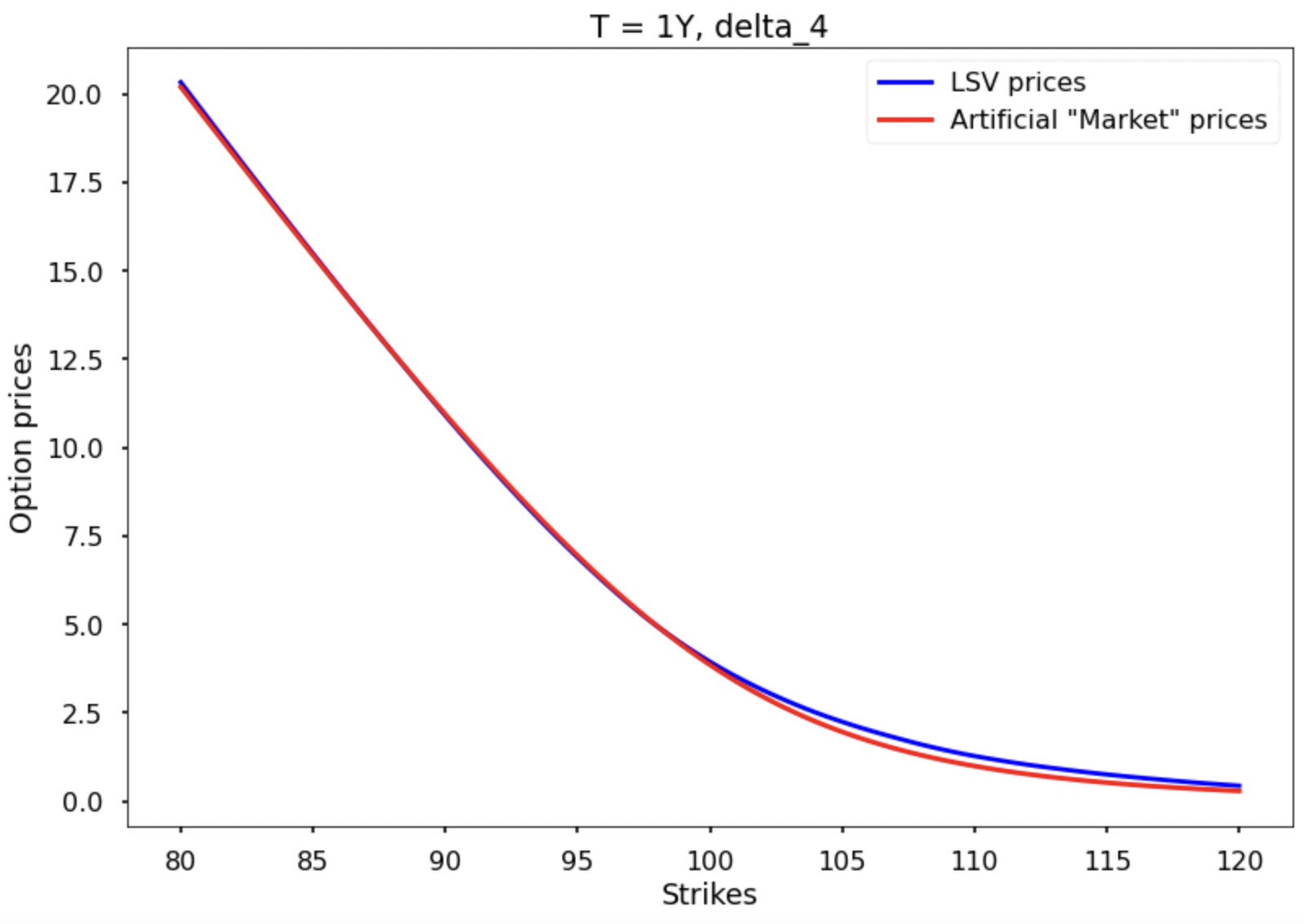}
     \caption{Prices comparison, $\delta_4 = 0.001\cdot \delta_1$}\label{Fig:d2}
   \end{minipage}
\end{figure}

Finally, we investigate how the choice of $\epsilon$ affects the convergence of the Euler-Maruyama scheme and the pathwise strong propagation of chaos. Specifically, we test the convergence results for the bandwidths: ${\epsilon_1 = 10,\epsilon_2 = 0.1, \epsilon_3 = 0.001}.$\newline
In Figure \ref{Fig:euler1}, we observe the strong convergence of the discretised scheme with a rate of order $1/2$ in the time\textendash step as expected by theory and proved in Section \ref{chap:euler} above. We notice that all bandwidths give an accurate result.
\begin{figure}[H]
   \begin{minipage}{0.55\textwidth}
     \includegraphics[width=.8\linewidth]{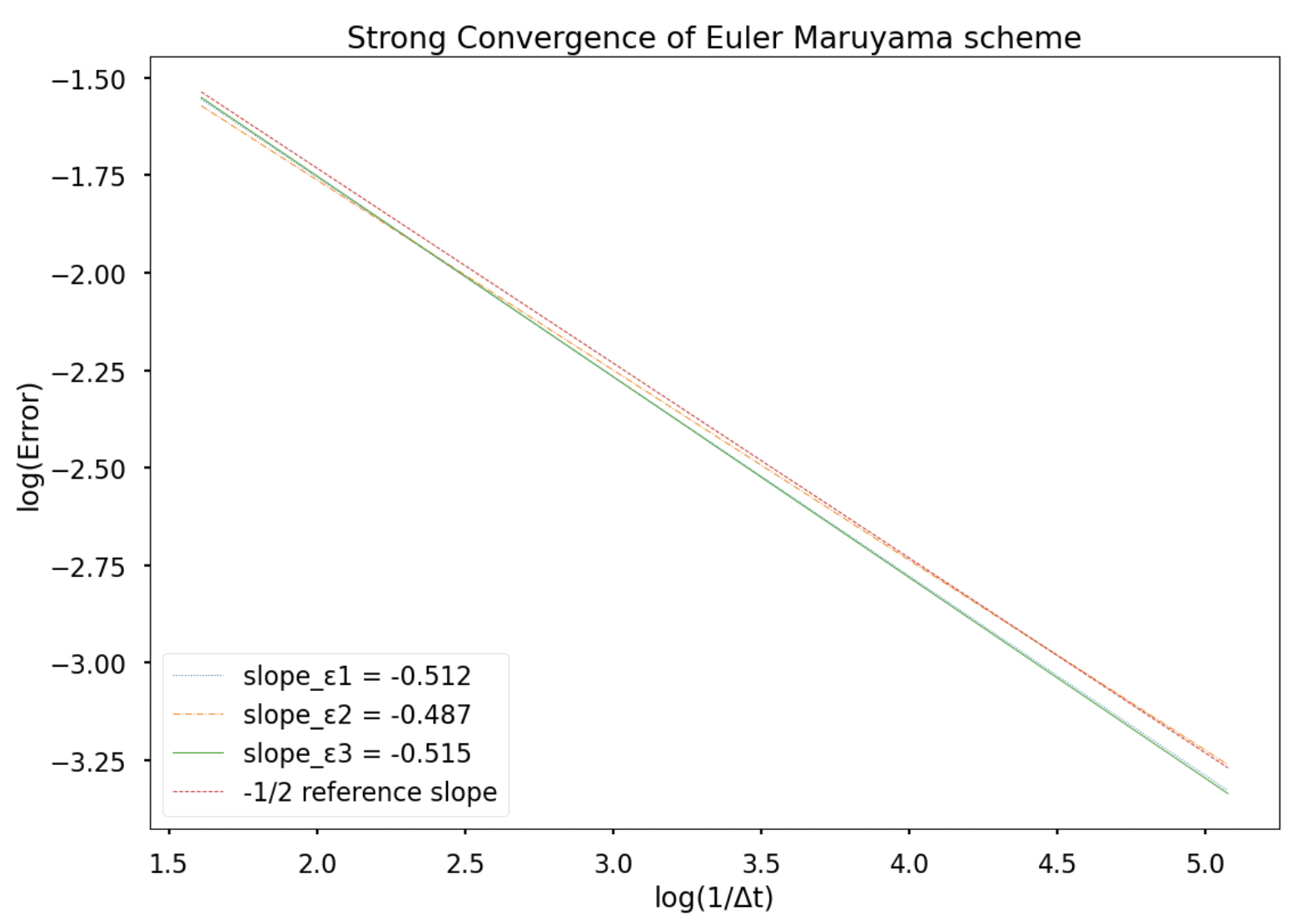}
     \caption{Strong convergence of \newline discretised scheme}\label{Fig:euler1}
   \end{minipage}\hfill
   \begin{minipage}{0.55\textwidth}
     \includegraphics[width=.8\linewidth]{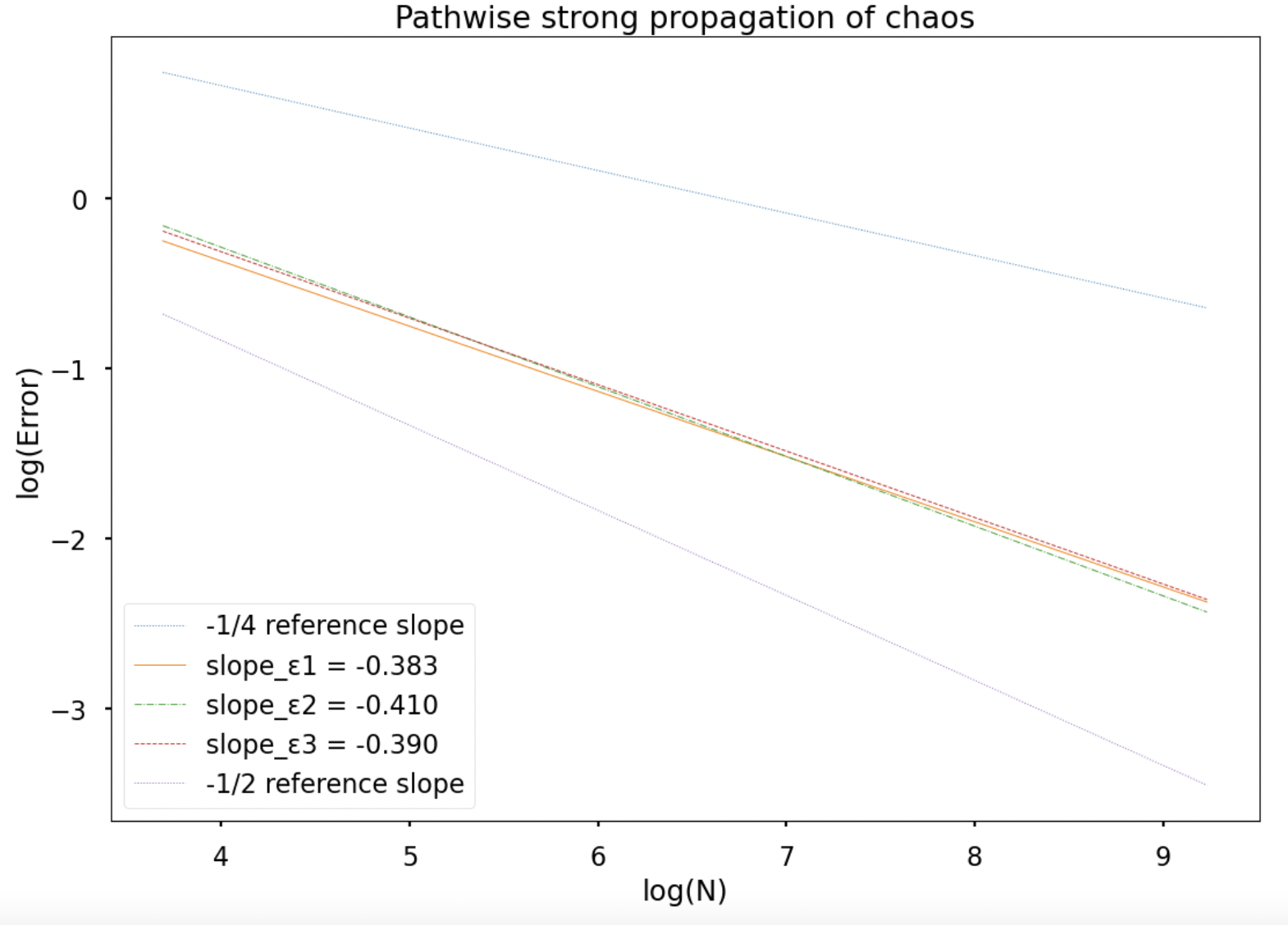}
     \caption{Pathwise strong convergence \newline of particle system}\label{Fig:particle2}
   \end{minipage}
\end{figure}
To illustrate the propagation of chaos result, in Figure \ref{Fig:particle2} we plot the following RMSE error across increasing $N$:
$ \text{error := }\sqrt{\frac{1}{2N}\sum^{2N}_{i=1}(S^{i,2N}_T - \tilde{S}^{i,2N}_T)^2},$
where both particle systems $\{{S}^{i,2N}_T\}_{i \in \{1,\dots, 2N\}}$ and  $\{\tilde{S}^{i,2N}_T\}_{i \in \{1,\dots, 2N\}}$ consist of $2N$ particles and use the same Brownian motions while for the particles ${\tilde{S}}^{i,2N}_T$, the leverage function is computed using only the first $N$ particles. 
Using $\epsilon = c*S_0N^{-1/5}$, we observe a strong convergence rate roughly of order $0.4$ in $N$ uniformly in $c$. 
The theoretical result in \cite{smith}, Proposition 3.1, gives an order of $1/4$ in the regular setting.

We conclude that a careful and well-studied choice of the regularisation parameters $\epsilon$ and $\delta$ is crucial since it significantly affects the accuracy of the calibration. 



\begin{thebibliography}{30}

\bibitem{abergeltachet} F. Abergel and R. Tachet, A nonlinear partial integrodifferential equation from mathematical finance,
\emph{Discrete and Continuous Dynamical Systems}, 27 (3): 907--917, 2010. 

\bibitem{BaoHuang} J. Bao and X. Huang, Approximations of McKean--Vlasov stochastic differential equations with irregular coefficients, \emph{Journal of Theoretical Probability}, 35, 1187--1215, 2022.

\bibitem{bayer} C. Bayer, D. Belomestny, O. Butkovsky, and J. Schoenmakers, RKHS regularization of singular local stochastic volatility McKean--Vlasov models, arXiv:2203.01160, 2022.

\bibitem{bossytalay} M. Bossy and D. Talay, A stochastic particle method for the McKean--Vlasov and the Burgers equation,
\emph{Mathematics of Computation}, 66 (217), 157--192, 1997.

\bibitem{hestoncal} A. Cozma, M. Mariapragassam, and C. Reisinger, Calibration of a hybrid local-stochastic volatility stochastic rates model with a control variate particle method, \emph{SIAM Journal on Financial Mathematics}, 10(1), 2019.

\bibitem{CUCHIERO} C. Cuchiero, W. Khosrawi, and J. Teichmann, A generative adversarial network approach to calibration of local stochastic volatility models, \emph{Risks}  8(4), 2020. 

\bibitem{newpaper} M.F. Djete, Non-regular McKean--Vlasov equations and calibration problem in local stochastic volatility models, arXiv:2208.09986, 2022.

\bibitem{smith} G. dos Reis, S. Engelhardt, and G. Smith, Simulation of McKean--Vlasov SDEs with super\textendash linear growth, \emph{IMA Journal of Numerical Analysis}, 42(1), 874--922, 2022.

\bibitem{well-posedness} G. dos Reis, W. Salkeld, and J. Tugaut, Freidlin--Wentzell LDPs in path space for McKean--Vlasov equations and the functional iterated logarithm law, \emph{Annals of Applied Probability}, 29(3), 2017.

\bibitem{DUPIREE} B. Dupire, Pricing with a smile, \emph{Risk}, 7, 18--20,1994. 

\bibitem{iff} B. Dupire, A Unified theory of volatility, In \emph{Derivatives Pricing: The Classic Collection}, edited by Peter Carr, Risk publications, 1996.

\bibitem{guo} I. Guo, G. Loeper, and S. Wang, Calibration of local-stochastic volatility models by optimal transport, \emph{Mathematical Finance}, 23(1), 2022.

\bibitem{mimicking} I. Gy{\"o}ngy, Mimicking the one-dimensional marginal distributions of processes having an It{\^o} differential, \emph{Probability Theory and Related Fields},71,501--516, 1986.

\bibitem{resolved} J. Guyon and P. Henry-Labord{\'e}re, Being particular about calibration, \emph{Risk}, 25(1), 92--97, 2012.

\bibitem{book} J. Guyon and P. Henry-Labord{\'e}re, Nonlinear Option Pricing, \emph{Chapman and Hall/CRC}, 2013.


\bibitem{jex} 
M. Jex, R. Henderson, and D. Wang, Pricing exotics under the smile, \emph{Risk}, 12, 72--75, 1999.

\bibitem{jourdain}
B. Jourdain and A. Zhou, Existence of a calibrated regime switching local volatility model, \emph{Mathematical Finance}, 30(2), 501--546, 2020.

\bibitem{AMISE} H. J. Kim, S. N. MacEachern, and Y. Jung, Bandwidth selection for kernel density estimation with a Markov Chain Monte Carlo sample, 
arXiv:1607.08274, 2016.

\bibitem{doksasi} P. E. Kloeden and E. Platen, Numerical Solution of Stochastic Differential Equations, \emph{Springer} Berlin, Stochastic Modelling and Applied Probability, Vol.\ 23, 1992.



\bibitem{lacker} D. Lacker, M. Shkolnikov, and J. Zhang. Inverting the Markovian projection, with an application to local stochastic volatility models.
\emph{Annals of Probability}, 48(5), 2189--2211, 2020.


\bibitem{lipton}
A. Lipton, The vol smile problem, \emph{Risk}, 15, 61--65, 2002.

\bibitem{Liu.Shi.Wu} H. Liu, B. Shi, and F. Wu, Tamed Euler--Maruyama approximation of McKean--Vlasov stochastic differential equations with super-linear drift and H{\"o}lder diffusion coefficients, \emph{Applied Numerical Mathematics}, 183, 56--85, 2023.

\bibitem{MKV} H.P. McKean, A class of Markov processes associated with nonlinear parabolic equations, \emph{Proceedings
of the National Academy of Sciences of the USA}, 56(6):1907--1911, 1996.


\bibitem{piterbarg}
V. Piterbarg, Markovian projection method for volatility calibration, SSRN 906473, 2006.


\bibitem{adaptive} C. Reisinger and W. Stockinger, An adaptive Euler--Maruyama scheme for McKean--Vlasov SDEs with super-linear growth and application to the mean-field FitzHugh--Nagumo model, \emph{Journal of Computational and Applied Mathematics}, 400, 2022.


\bibitem{ren} Y. Ren, D. Madan, and M. Qian Qian, Calibrating and pricing with embedded local volatility models, \emph{Risk}, 20, 138--143, 2007.


\bibitem{1991} A.S. Sznitman, Topics in propagation of chaos, Ecole d'Et{\'e} de probabilit{\'e}s de Saint-Flour XIX -- 1989, Lecture Notes in Mathematics, Vol.\ 1464, \emph{Springer}, Berlin, 1991.  

\bibitem{Villani} C. Villani, Optimal transport, \emph{Grundlehren der Mathematischen Wissenschaften}, Volume 338, Springer Berlin, 2009. 

\bibitem{X.Zhang} X. Zhang, A discretized version of Krylov's estimate and its applications, \emph{Electron. J. Probab.}, 24, 1--17, 2019.







\end{thebibliography}
\end{document}